\documentclass[onecoloumn]{IEEEtran}
\usepackage[utf8]{inputenc}
\usepackage{nomencl}
\usepackage{cite}
\usepackage[dvips]{graphicx,psfrag}
\usepackage{amsthm}
\usepackage{amsmath}
\usepackage{amsmath}
\usepackage{amssymb}
\usepackage{rotating}
\usepackage{subfigure}
\usepackage{soul}
\usepackage{cite}
\usepackage{color}
\usepackage{stfloats}
\usepackage{pgf,tikz}
\usepackage{mathrsfs}
\usetikzlibrary{arrows}
\newtheorem{theorem}{Theorem}
\newtheorem{lemma}{Lemma}
\newtheorem{corollary}{Corollary}

\usepackage{hyperref}
\usepackage{algorithmic}
\usepackage[linesnumbered,ruled,vlined]{algorithm2e}
\usepackage{mathtools}
\allowdisplaybreaks

\newcommand\eqa{\mathrel{\overset{\makebox[0pt]{\mbox{\normalfont\tiny\sffamily (a)}}}{=}}}
\newcommand\eqb{\mathrel{\overset{\makebox[0pt]{\mbox{\normalfont\tiny\sffamily (b)}}}{=}}}
\newcommand\eqc{\mathrel{\overset{\makebox[0pt]{\mbox{\normalfont\tiny\sffamily (c)}}}{=}}}

\newcommand\eqf{\mathrel{\overset{\makebox[0pt]{\mbox{\normalfont\tiny\sffamily (f)}}}{=}}}

\newcommand\leqc{\stackrel{\mathclap{\normalfont\mbox{\normalfont\tiny\sffamily (c)}}}{\leq}}
\newcommand\leqd{\stackrel{\mathclap{\normalfont\mbox{\normalfont\tiny\sffamily (d)}}}{\leq}}

\newcommand\approxa{\stackrel{\mathclap{\normalfont\mbox{\normalfont\tiny\sffamily (a)}}}{\approx}}

\newcommand\approxe{\stackrel{\mathclap{\normalfont\mbox{\normalfont\tiny\sffamily (e)}}}{\approx}}

\begin{document}
\definecolor{ffffff}{rgb}{1.,1.,1.}
\definecolor{ffqqqq}{rgb}{1.,0.,0.}
\definecolor{qqqqff}{rgb}{0.,0.,1.}
\definecolor{uuuuuu}{rgb}{0.26666666666666666,0.26666666666666666,0.26666666666666666}
\definecolor{zzttqq}{rgb}{0.6,0.2,0.}
\definecolor{xdxdff}{rgb}{0.49019607843137253,0.49019607843137253,1.}
\definecolor{ffqqff}{rgb}{1.,0.,1.}
\definecolor{cqcqcq}{rgb}{0.7529411764705882,0.7529411764705882,0.7529411764705882}
\title{Coverage Analysis of Two-Tier HetNets for Co-Channel, Orthogonal, and Partial Spectrum Sharing under Fractional Load Conditions}

\author{Praful D. Mankar, Goutam Das,~\IEEEmembership{Member,~IEEE} and S. S. Pathak,~\IEEEmembership{Senior Member,~IEEE}
\thanks{Copyright (c) 2015 IEEE. Personal use of this material is
permitted. However, permission to use this material for any other purposes must be obtained from the
IEEE by sending a request to pubs-permissions@ieee.org.}
\thanks{Praful D. Mankar and S. S. Pathak are with the Dept. of E \& ECE; and Goutam Das is with the G. S. Sanyal School of Telecommunications, Indian Institute of Technology, Kharagpur, WB, India. E-mail: \{praful,ssp\}@ece.iitkgp.ernet.in and gdas@gssst.iitkgp.ernet.in}
}

\maketitle
\begin{abstract}
In heterogeneous networks, the random deployment of femto access points (FAPs) and macro base stations (MBSs) with uncoordinated channel access impose huge inter-tier interferences. In real-life networks, the process of MBSs' deployment exhibits the homogeneity, however the FAPs have the behavioral characteristic of clusters formation like in  malls, apartments, offices, etc. Therefore, the composite modeling of the MBSs and the FAPs using Poisson point process and Poisson cluster process is employed for the evaluation of coverage probability. The scenario of the real-time traffic for macro-tier and the best-effort traffic for femto-tier is considered. Cognition is introduced in the clustered FAPs to control the inter-tier interference. Furthermore, the impact of macro-tier load is analyzed by exploiting the inherent coupling between coverage probability and activity factor of an MBS. Further, we study the effect of co-channel, orthogonal, and partial spectrum sharing modes on the coverage for given parameters like load condition, FAPs/MBSs density, etc. We provide simulation validation for the derived expressions of coverage and present an comparative analysis for the mentioned spectrum sharing modes.
\end{abstract}
\IEEEpeerreviewmaketitle
\section{Introduction}
\IEEEPARstart{T}{o} expedite support for high data rate services, the concept of subscriber-owned small base stations (BSs) (such as femto cells), beside the tiers of macro/pico cells, is being adopted by the next generation cellular standards such as LTE and WiMax. However, the random deployment of femto access points (FAPs) and their random channel access pattern introduce huge interference on the users associated with macro-tier causing a serious degradation of the overall network capacity. To minimize cross-tier interference, prior research has proposed techniques like hybrid spectrum sharing, transmit power control, and adaptive channel access, etc. The concern of this paper is to motivate and propose a method to select an effective spectrum sharing mode for the deployment of two-tier networks. We also introduce cognitive channel access at the small cells level to further control inter-tier interference that requires local information instead of at at macro cell level where global information might be essential.  

The cellular BSs are deployed in accordance with the traffic demand which is often independent from one geographical region to another. Therefor, the fact of non-uniform deployment has motivated the researcher to model the BSs using stochastic geometry models. 
The most natural choice of point process for modeling the BSs locations is homogeneous Poisson point process (PPP) \cite{Haenggi2009,haenggi2009interference,Baccelli_stochasticgeometry,Stoyan,Hossain_2013,Andrews_2011,JXu_2011}.  
The homogeneous PPP promises to be the right convention for modeling the cellular BSs since their deployment is bounded to yield ubiquitous connectivity. In other words, the homogeneous PPP is more accurate for modeling the cells having different coverage range as the service demand across the region is non-uniform. 
Moreover, the application of PPP renders mathematical tractability in the analysis along with practical relevance. In \cite{Andrews_2011}, the authors have shown that the assumption of PPP provides a pessimistic model for coverage/rate with  almost equal deviation as the optimistic model resultant of grid-based modeling.
The authors of \cite{WeberS_2010} have derived exact form of upper and lower bounds on the coverage probability, and further provide insights into the achievable transmission capacity. 

Furthermore, the researchers have extended the stochastic geometry approach to analyze the performance  of multi-tier network wherein the tiers of BSs are modeled using independent homogeneous PPPs. The authors of \cite{Dhillon_2012,Dhillon_2013,Dhillon_2014,HSJo_2012,HSJo_2011,SSingh_Offloading} have extended analysis of \cite{Andrews_2011} to multi-tier HetNet. Further the author of \cite{Heath_2013,Renzo,Dhillon_2014} have modeled the interference in such networks under generalized fading. 
On the contrary, assumption of modeling small cells using homogeneous PPP may not be accurate for modeling the subscriber-owned BSs such as FAPs which are added in the network on ad-hoc basis. The FAPs are more likely to form clusters as they are randomly plugged within residential complexes, malls, offices etc. 
 Therefore, to evaluate the performance of wireless nodes having behavioral characteristic of forming clusters are modeled using Poisson cluster process (PCP) in the literature.
In \cite{RKGanti_2009}, the expressions (in integral form) are derived for interference and outage for PCP modeled ad-hoc networks.  Authors of \cite{Haenggi_2012} realized the resemblance of process of forbidding the points in PPP to lie closer than a certain minimum distance (as done in cognitive networks) with the PCP. 
Thenceforth, using the analysis of \cite{RKGanti_2009}, the Laplace transform of interference for cognitive network is obtained. 
The authors of \cite{Afshang_2015,Afshang_2015a} have modeled the Device-to-device networks using PCP. The derived coverage probability is further used to characterize the area spectral efficiency of the network.  
Moreover, the authors of \cite{YJChun_2015} have extended analysis of \cite{RKGanti_2009} to investigate the outage and the achievable average rate in a  multi-tier HetNet by assuming that the BSs in each tier form an independent PCP. However, this assumption may not be practical for the case of macro-tier because of homogeneity in the deployment of macro base stations (MBSs). 
The authors of \cite{Zhong_2013,Luo} have realized that the cluster process is likely to be more realistic for the modeling of femto-tier as the FAPs are typically deployed in populous locations, like commercial or residential areas. Thus the composite modeling of MBSs and FAPs using PPP and PCP, respectively, is employed for the analysis in \cite{Zhong_2013,Luo}.

In all the above mentioned investigations, the interference analysis is undertaken for the network with only one frequency channel. However, the consideration of multiple channels and fractional load conditions may significantly change the interference realization in a single/multi-tier cellular network. 
Authors of \cite{Dhillon_RandomArrival} characterize the system throughput as a function of service arrival rate within the framework of uniformly distributed service in a single cell of fixed coverage. The analysis is conducted under the assumption of single time-frequency slot allocation.
The notion of load awareness in interference modeling in a multi-tier network is introduced in \cite{Dhillon_2013}. Therein fractional load of $i$-th tier is modeled as $p_i$ that represents the activity factor (represent the probability of a channel access) of a BS in $i$-th tier and is assumed to be known a priori. 
{The activity factor of a BS in PPP modeled network for elastic traffic is derived in \cite{Wei_Bao_2014_NearOptimal}. 
However, \cite{Wei_Bao_2014_NearOptimal} reserves unit bandwidth per service arrival which is not the case in actual networks. The bandwidth consumption of a service is dependent on the SINR statistics that is further dependent on location of its arrival.}
The consideration of multiple channels and fractional load conditions may significantly change the interference realization in a single/multi-tier cellular network.
In fact there is a implicit coupling between the activity factor and the coverage probability. For example transmission of a BS, say B1, generates interference to its neighboring BS, say B2, which force B2 to transmit for longer time which again interferes back to B1 and make him transmit for even longer time. 
Thus, it is essential to resolve the coupling between activity factor and coverage probability to investigate the load-aware performance of such networks.
Further, to understand the impact of underlying traffic load on the BS activity factor in a multi-channel environment, it becomes necessary to incorporate stochastic traffic into the activity factor evaluation.

In \cite{EkramHossain_2013,EkramHossain_2014} the two-tier network comprised of macro cells and femto cells sharing the same set of channels are modeled using independent PPPs. To avoid the inter-tier interference, the femto cell is assumed to access a channel only if the received signal strength of MBSs is below a sensing threshold. The orthogonal spectrum allocation techniques in open (closed) access PPP modeled multi tier cellular network are presented in \cite{Wei_Bao_2014_NearOptimal, Wei_Bao_2014_Structured, Adve_2014} (\cite{Chandrasekhar_2009}) such that the area spectral efficiency is optimum in full load scenario. Further, the authors of \cite{Wei_Bao_2014_Structured,Adve_2014} have also considered tier biasing as the optimization parameter as it has direct impact on the mean rate of a user in open access scenario. Moreover, in \cite{YoungjuKim_2010} the area spectral efficiency in closed access two-tier network for partial spectrum sharing is investigated. 
However, co-channel/partial deployment of multi-tier network might be effective over orthogonal deployment for a particular range of traffic intensity, BSs densities, and power configuration. Therefore, it is essential to perform a comparative analysis among the co-channel, partial, and orthogonal deployments.   

It is clear that the network analysis under PPP modeled macro-tier and PCP modeled femto-tier has more practical relevance compared to modeling both tiers using either  PPP or PCP. In the existing literature, the coverage analysis is conducted for the best-effort traffic only except \cite{Wei_Bao_2014_NearOptimal, Dhillon_2013,Dhillon_RandomArrival}. {However, the influence of real-time service on network performance is still at the early stage under the paradigm of random cellular network modeling.}
Therefore, investigation of coverage probability for such two-tier HetNets in multi-channel environment under fractional load condition along with comparative analysis of co-channel, partial, and orthogonal deployment modes should be more interesting and practically relevant. Therefore in this paper, we set out for such investigation. In the following section we enlist the contributions of this paper.
\subsection{Contributions}
The earlier mentioned practicability reasons have motivated us to analyze coverage probability and rate distribution for the case of PPP modeled macro-tier and PCP modeled femto-tier in a multi-channel environment. 
To control inter-tier interference issue, we have introduced cognition in clustered FAPs. 
The considered underlay mode of cognition allow FAPs in a cluster to access a channel for transmission only if the signal sensing level from the nearest MU is below some threshold, or equivalently if the nearest MU is beyond exclusion range. 
{Further in this paper, we have considered more realistic scenario of the real-time traffic for macro-tier and  best-effort traffic for femto-tier under closed access}. 
Due to the randomness of service arrival and departure process, activity factor of an MBS becomes a function of traffic intensity and coverage probability.
Therefore, exploiting the implicit coupling of coverage and activity factor, while relaxing the fixed bandwidth allocation as in \cite{Wei_Bao_2014_NearOptimal,Dhillon_RandomArrival}, we incorporate a fractional load scenario of real-time traffic into the analysis.  We have derived the load-aware coverage probabilities for macro users (MUs) and femto users (FUs) in a two-tier network with clustered cognitive FAPs. 

Furthermore, properly configured spectrum sharing is one of the potential methods to enhance the coverage/rate in a closed accessed HetNets.
Therefore, we have also investigated the coverage analysis in three types of spectrum sharing modes \textit{viz.}: co-channel, orthogonal, and partial. 
In this paper, the co-channel deployment of multi-tier network is proven to be yield better coverage it the certain density of FAPs and beyond it the orthogonal mode becomes dominant choice of deployment.
Through numerical results, we show that the best choice of spectrum sharing mode is dependent on given system parameters like macro cell density, femto cell density, traffic load, etc.   
Further, we provide comprehensive analysis for the choice of performance assessment parameters, like exclusion range and frequency partitioning, along with the comparison of the spectrum sharing modes for a given system parameters. The contributions of the paper are summarized in the following:
   \begin{itemize}
    \item The load-aware performance analysis of two-tier HetNet with cognitive clustered FAPs. 
    \begin{itemize}     
    \item A framework is developed to resolve the coupling between coverage probability and activity factor.    
    \item A simplistic approximated model for the activity factor of an MBS is presented.    
    \end{itemize}    
    \item Performance analysis is presented for three spectrum allocation methods viz. co-channel, orthogonal, and partial. 
    \begin{itemize}
     \item It is shown that the co-channel and orthogonal modes are effective for the FAPs density less than and greater than a critical point, respectively. 
    \end{itemize}    
    \item The numerical analysis indicates the mode of deployment (i.e. type of spectrum allocation method, exclusion range, etc.) is dependent of network parameters such as macro-tire load, density of FAPs, density of MBSs, etc.
   \end{itemize}
   
The rest of the paper is organized as follows. Section \ref{sec:SystmModelAndAssumption} presents the system model and assumptions. The load-aware coverage analysis in  under a co-channel, orthogonal and partial modes is carried out in Section \ref{sec:CoverageAnalysisforCo-ChannelAllocation}, Section \ref{sec:CoverageAnalysisforOrthogonalSpectrumAllocation}, and Section \ref{sec:CoverageAnalysisforPartialSpectrumAllocation}, respectively. Next, Section \ref{sec:FurtherAnalysis} provides few key insights into the network design.
In Section \ref{sec:NumericalResultsAndDiscussion}, we present numerical results to demonstrate the coverage performance of users of both tiers along with a comparative analysis of coverage probability in different spectrum sharing modes. Finally, we conclude the paper in Section \ref{sec:Conclusion}.
\section{System Model and Assumptions}
\label{sec:SystmModelAndAssumption}
We consider a two-tier HetNet comprised of macro cells and femto cells. Assuming that the both tiers are independent, we model MBSs location using homogeneous PPP \mbox{\small{$\Phi_B$}} with intensity \mbox{\small{$\lambda_B$}} in \mbox{\small{$\mathbb{R}^2$}} \cite{Andrews_2011,haenggi2009interference,Dhillon_2012} and FAPs location using PCP \mbox{\small{$\Phi_F$}} in \mbox{\small{$\mathbb{R}^2$}} \cite{RKGanti_2009,YJChun_2015}.  
Applying homogeneous independent clustering to the points of a homogeneous PPP results in PCP. It is also referred as Neyman-Scott cluster process. Let $\lambda_F$ be the density of PPP $\Phi_p$ that forms the parents points of the clusters and $c$ be the mean number of points (FAPs) in a cluster.  Therefore, the density of PCP $\Phi_F$ becomes $\lambda_Fc$. Let the cluster be of the form $\phi_c^x=\{y+x|y\in\phi_c\}$ for each $x\in\Phi_p$ and $\phi_c\in\mathcal{F}$ where $\mathcal{F}$ represent the family of random clusters of points centered at origin.
The complete PCP $\Phi_F$ becomes\begingroup\makeatletter\def\f@size{8}\check@mathfonts
\begin{equation}
\Phi_F=\bigcup\limits_{x\in\Phi_p}\phi_c^x 
\end{equation}\endgroup
There are two processes to realize a cluster $\phi_c$ of PCP, \textit{viz.}: Matern cluster process and Thomas cluster process. In both of these processes number of points in a cluster are Poisson distributed with mean $c$. In Matern process, each point is uniformly distributed over a disk of radius $R$  around origin. However, in Thomas process, each point is distributed using a symmetric normal distribution around the origin. For sake of simplicity we consider Matern cluster process $\phi_c$ to realize the clusters of the PCP $\Phi_F$.  Therefore, the distribution of a FAP in a cluster centered at origin can be written as \begingroup\makeatletter\def\f@size{8}\check@mathfonts
\begin{equation}
 f(y)=\begin{cases}
       &\frac{1}{\pi R^2},~~~~\|y\|\leq R\\
       &0, ~~~~~~~~\text{otherwise}
      \end{cases}
\end{equation}\endgroup

Both tiers use the set \mbox{\small{$\mathcal N$}} of \mbox{\small{$N$}} number of channels each of bandwidth \mbox{\small{$B$}}. The transmission powers of MBSs and FAPs are represented by \mbox{\small{$P_B$}} and \mbox{\small{$P_F$}}, respectively. Note that no power control is assumed at the BSs of either of the tiers. In our analysis we consider macro-tier and femto-tier to have distinct set of users, i.e. a closed access scenario is considered. Further, let \mbox{\small{$\mathcal{M}$}} be the set of \mbox{\small{$T$}} modulation and coding schemes (MCSs) employed to transmit data to the users. An interference limited scenario is considered and therefore, the noise is ignored in the analysis. Small scale fading is assumed, where the fading coefficients are assumed to be exponentially distributed with unity mean. A general power-law path loss model is used, i.e. the path loss with distance \mbox{\small{$r$}} is modeled as \mbox{\small{$r^{-\alpha}$}}, where \mbox{\small{$\alpha$}} is a path loss exponent. Let $\delta=\frac{2}{\alpha}$. Although, the path loss exponent may be a tier specific parameter, however like \cite{EkramHossain_2014}, the same path loss exponent for both tiers is considered for analytical tractability.
\subsection{Spectrum Sharing Modes for Two-tier Networks}
The spectrum sharing mode selection plays an important role in realization of coverage, mean user throughput, achievable network capacity, etc. 
In literature, mainly three types of spectrum sharing in a two tired network are available \textit{viz.}: co-channel, orthogonal, and partial. In co-channel spectrum sharing mode the FABs and MBSs have access to the same set of channels. However, in orthogonal spectrum sharing mode, FAPs and MBSs have distinct sets of channels. 
A subset of channels is accessed by both tiers and the remaining subset of channels is reserved for macro-tier only under the assumption of partial mode. The partial mode basically evades the fraction of interference to macro user from FAPs by limiting the access of set of channels for FAPs. 
\subsection{Macro User SIR and Channel Access Model under Co-Channel Mode}
\label{sec:MacroUserSINRandChannelAccessModel}
The MU is assumed to be associated to the BS with maximum received power. Let the location of serving MBS be \mbox{\small{$x_0$}}, i.e. \mbox{\small{$x_0=\arg\max_{x_i\in\Phi_B}P_B\|x_i\|^{-\alpha}$}}.
Without loss of generality, the MU is assumed to be located at the origin. 
The signal-to-interference ratio (SIR) of an MU associated to the MBS at \mbox{\small{$x_0\in\Phi_B$ ($\|x_0\|=r$)}} can be written as 
\begingroup\makeatletter\def\f@size{8}\check@mathfonts
\begin{equation}
 \Gamma_M(r) = \frac{h_0r^{-\alpha}P_B}{I_{\text{mm}}(r) + I_{\text{fm}}}
 \label{eq:MacroUserSINR}
\end{equation}\endgroup
where \newline\mbox{\small{$I_{\text{mm}}(r)=\sum_{x_i\in\Phi_B\setminus x_0} h_i\|x_i\|^{-\alpha}\xi_iP_B$}} is MBSs interference, 
\newline\mbox{\small{$I_{\text{fm}}=\sum_{x_i\in\tilde\Phi_F}h_i\|x_i\|^{-\alpha}\chi {P_F}$}} is FAPs interference. 
The \mbox{\small{$\tilde\Phi_F$}} represent the set of interfering co-channel FAPs in \mbox{\small{$\Phi_F$}}. 
The $\xi_i$'s are Bernoulli random variable with parameter \mbox{\small{$\zeta$}} that represents the activity factor of an MBS. We represent the activity factor for co-channel mode by $\zeta_C$, for orthogonal mode by \mbox{\small{$\zeta_O$}}, and for partial mode by \mbox{\small{$\zeta_P$}}.
The \mbox{\small{$h_i$'s}} represent the Rayleigh fading co-efficients and are i.i.d exponential random variable with unit mean. The \mbox{\small{$\chi$}} denotes the wall penetration loss. 

{The macro-tier is considered to provide real-time service having rate requirement of \mbox{\small{$R_{\text{th}}$}}.
The cellular traffic model requires the description of the service arrival and departure process along the time axis that by definition are stochastic in nature. Therefore, traffic modeling in random networks (in Euclidean space) will append the time dimension to the point process. Hence, the overall analysis requires modeling of a three-dimensional stochastic process where the arrival is random in both time and space. For this reason, the service/user arrival process is modeled using space-time PPP (STPPP) of intensity $\lambda_\text{M}$ (units/min$\cdot$m$^2$) in $\mathbb{R}^3$. Note that the traffic intensity $\lambda_\text{M}$ specifies the network load.
The user stays in the network for exponentially distributed time of mean $1/\mu$.
\begin{figure}
\centering
\hspace{-.45cm}\includegraphics[width=.26\textwidth]{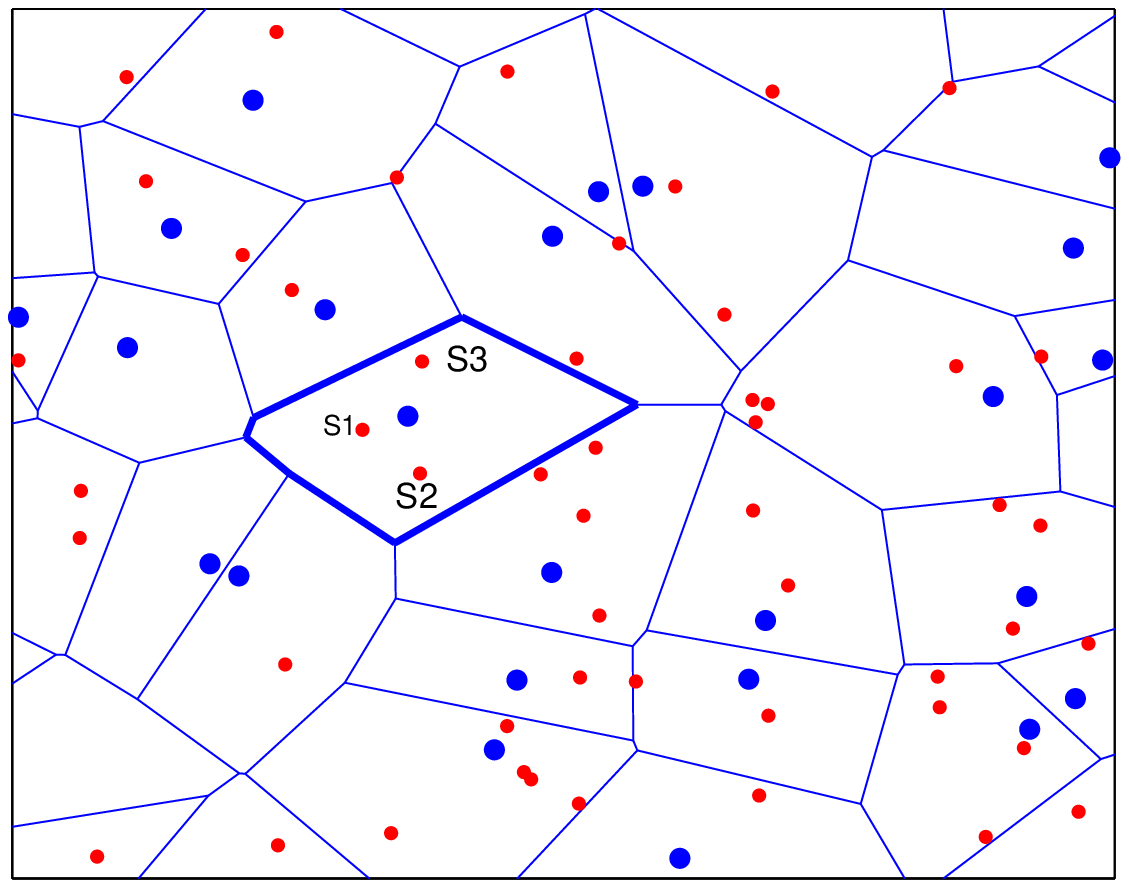}\hspace{-.35cm}\includegraphics[width=.26\textwidth]{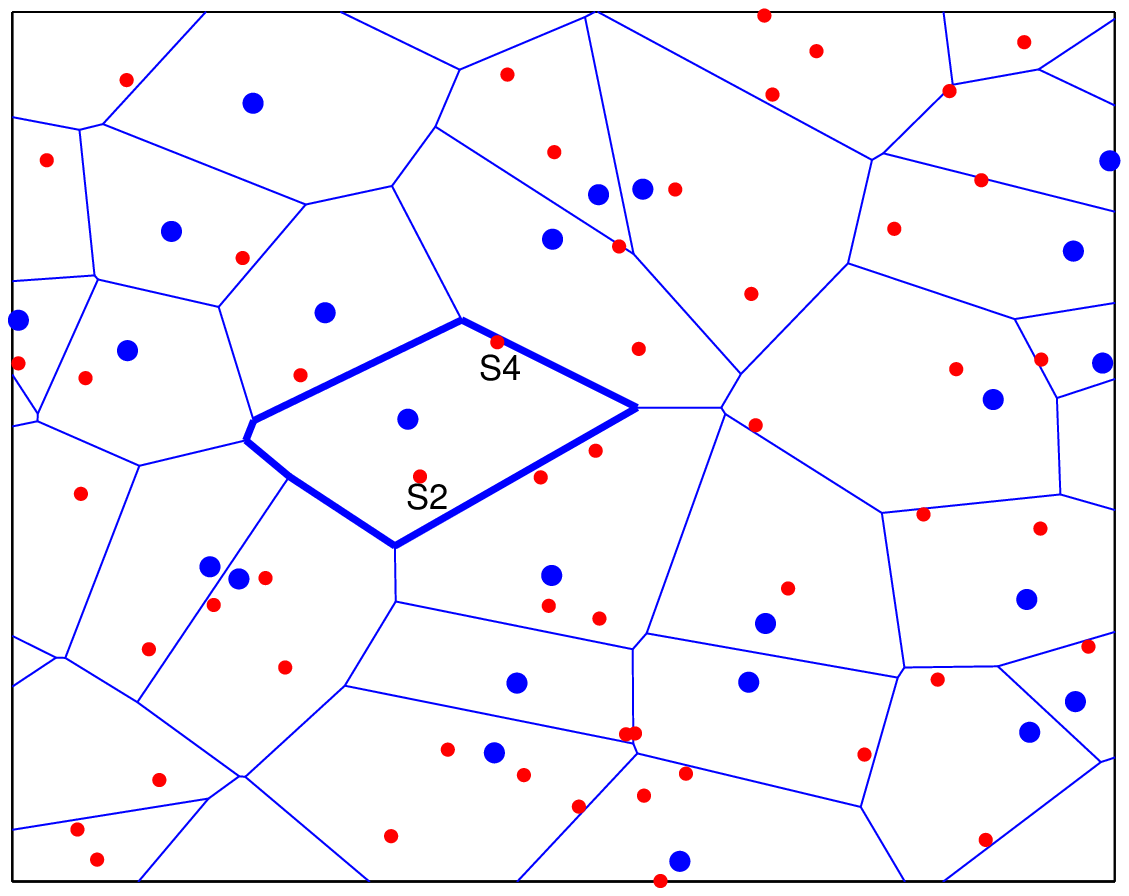}  
\vspace{-.15cm}
\resizebox{.33\textwidth}{!}{ 
 \begin{tikzpicture}[line cap=round,line join=round,>=triangle 45,x=1.0cm,y=1.0cm]
\clip(4.,1.) rectangle (15.,6.);
\draw [->] (4.54,1.42) -- (4.56,5.28);
\draw [->] (4.22,1.72) -- (14.88,1.72);
\draw (6.04,2.28)-- (6.04,1.72);
\draw (6.04,2.28)-- (9.9,2.28);
\draw (9.9,2.28)-- (9.9,1.72);
\draw (7.,3.2)-- (14.,3.2);
\draw (14.,3.2)-- (14.,2.6);
\draw (14.,2.6)-- (7.,2.6);
\draw (7.,2.6)-- (7.,3.2);
\draw (7.6,3.6)-- (7.6,4.2);
\draw (7.6,4.2)-- (10.8,4.2);
\draw (10.8,4.2)-- (10.8,3.6);
\draw (10.8,3.6)-- (7.6,3.6);
\draw (10.5,5.2)-- (14.3,5.2);
\draw (14.3,5.2)-- (14.3,4.6);
\draw (14.3,4.6)-- (10.5,4.6);
\draw (10.5,4.6)-- (10.5,5.2);
\draw (7.4,2.2) node[anchor=north west] {S1};
\draw (10.14,3.12) node[anchor=north west] {S2};
\draw (8.96,4.14) node[anchor=north west] {S3};
\draw (12.1,5.14) node[anchor=north west] {S4};
\draw (13.86,1.46) node[anchor=north west] {Time};
\draw [dash pattern=on 2pt off 2pt] (12.66,1.36)-- (12.66,5.54);
\draw (11.45,1.64) node[anchor=north west] {$t+\Delta t$};
\draw [dash pattern=on 2pt off 2pt] (8.26,1.14)-- (8.26,5.32);
\draw (7.85,1.64) node[anchor=north west] {$t$};
\end{tikzpicture}}
\vspace{-.15cm}
\caption{Illustration of STPPP for macro tier traffic. Top: snapshot of macro tier at time $t$ (top left) and $t+\Delta t$ (top right) [Blue dot denotes MBS and red dot denoted the user/service location]. Bottom: services in process of highlighted cell at time instances $t$ and $t+\Delta t$.}
\label{fig:STPPP}
\vspace{-.4cm}
\end{figure} 
Fig. \ref{fig:STPPP} illustrates the STPPP modeling of cellular traffic. Snapshots of the network, comprised of users and MBS, at time $t$ and $t+\Delta t$ are depicted in the figure. It can be seen that the locations of MBS, modeled using PPP $\Phi_B$ with density $\lambda_B$ (units/m$^2$) in $\mathbb{R}^2$, are static and do not change with time. On the other hand, the set of users changes over time as users randomly arrives in and departs from the network.  For example, consider the highlighted cell with incumbent service arrival and departure instances as shown in the figure. It can be seen that the cell comprises of three services  (say S1, S2, and S3) at time $t$. By time $t+\Delta t$, services S1 and S3 departs the cell
and new service S4 joins the cell. 
Ignoring the call blocking, traffic of each cell can be modeled using independent $M/M/\infty$ queues with different arrival rates depending on cell size. Thus the number of users in a typical cell becomes a Poisson random variable with parameter $\frac{\lambda_Ma}{\mu}$ \cite{Kleinrock} where $a$ is the cell area. Therefore, as the cells are disjoint with time invariant areas and users are uniformly distributed in each cell, the resultant process of user form a PPP of intensity $\frac{\lambda_M}{\mu}$ at any given time instance.

It is assumed that an MBS uniformly distribute its load across $N$ channels. Thus, the probability of an MU accessing  a channel becomes $\frac{1}{N}$. Therefore, the point process of MUs per channel can be obtained by splitting the overall PPP of macro service into $N$ processes independently with probability $\frac{1}{N}$. Therefore, MUs per channel can be modeled using homogeneous PPP $\Phi_m$ of density $\frac{\lambda_m}{\mu}=\frac{\lambda_M}{\mu}\frac{1}{N}$. Note that as a consequence, the intra-cell MUs  in  $\Phi_m$ are Poisson distributed with parameter $\frac{\lambda_ma}{\mu}$ where $a$ is the area of the cell. Therefore, it is assumed that these MUs access the channel at orthogonal time instances. 
\subsection{Femto User SIR and Channel Access Model under Co-Channel Mode}
\label{sec:FemtoUserSINRandChannelAccessModel}
The coverage analysis is carried out for an FU situated at the center of a cluster. 
The distance between the FAP and its user is assumed to be fixed (\mbox{\small{$r_0$}}) as in \cite{Chandrasekhar_2009} and the FAP is located at \mbox{\small{$x_1\in\phi_c$}}. The SIR of FU, located at origin, can be written as
\begingroup\makeatletter\def\f@size{8}\check@mathfonts
\begin{equation}
 \Gamma_F = \frac{h_0r_0^{-\alpha}P_F}{ I_{\text{mf}} + I_{\text{if}} + I_{\text{of}}}
 \label{eq:FemtoUserSINR}
\end{equation}\endgroup
where,
\newline \mbox{\small{$I_{\text{mf}}=\sum_{x_j\in \Phi_B}  h_j\|x_j\|^{-\alpha}\chi \xi_iP_B$}} is MBSs interference,
\newline \mbox{\small{$I_{\text{if}}=\sum_{x_j\in \tilde\phi_c\setminus x_1} h_j\|x_j\|^{-\alpha}\chi P_F$}} is intra-cluster interference, 
\newline \mbox{\small{$I_{\text{of}}=\sum_{x_j\in \tilde\Phi_F\setminus\tilde\phi_c} h_j\|x_j\|^{-\alpha}\chi P_F$}} is inter-cluster interference.
The $\tilde\phi_c$ is the set of interfering co-channel FAPs in $\phi_c$.

At any given time instance, the point process of macro users per channel follows PPP $\Phi_m$ with density $\frac{\lambda_m}{\mu}=\frac{\lambda_M}{N\mu}$ where overall $\lambda_M$ is the arrival density. It is assumed that the intra cell macro users in the PPP $\Phi_m$ accesses the channel in orthogonal time instances. The channel is considered to be accessible by the cognitive FAPs in a cluster only if the nearest MU in $\Phi_m$ (point process of the same channel) is beyond $r_m$ distance (i.e. exclusion range) from the center of the cluster. The probability that an MU in $\Phi_m$ is beyond $r_m$ can be written as \cite{HSJo_2012}
   \begin{equation}
   p_{r_m}=\exp\left(-\frac{\lambda_m}{\mu}\pi r_m^2\right) 
   \label{eq:ChannelAccessProb}
   \end{equation}
Therefore, as the MUs are modeled using $N$ independent PPPs, probability that a total of $n$ channels are accessible follows a Binomial distribution such that
\begin{align}
\mathbb{P}\left(n\right)&=\dbinom{N}{n} p_{r_m}^{n}\left(1-p_{r_m}\right)^{N-n}
\label{eq:nChannelAccessibleProb}
\end{align}
It is assumed that the accessible channels in a cluster are assigned such that each channel has the same number of co-channel FAPs.
This assumption capture the effect of the intra-cluster interference for given MU density. Thus effective number of co-channel FAPs increases with increase in MU density as $p_{r_m}$ varies inversely with MU density. Further, it is assumed that an FAP always have packets in buffer to transmit. Thus the FAP is always transmitting at achievable rate  of the channel. This type of traffic is referred as full-buffer best-effort traffic. This can  incorporated in the analysis by substituting the activity factor of FAPs equal to one. 

The typical example of the considered system model with opportunistic channel access for clustered FAPs is illustrated in Fig. \ref{fig:illustration}. It may be noted that the solid lines indicate desired links and dotted lines indicate interfering links.
For simplicity, single MU per channel under single macro cell is depicted. However, in actual we have assumed random sets of MUs per channel such that the channel is accessed by these MUs at orthogonal time instances. The node color indicates the channel index in which the wireless link is established. From the figure it can be seen that the nearest cluster of co-channel FAPs is alway beyond exclusion region of radius $r_m$. This makes void around an MU within which co-channel FAPs cannot exist. It can be seen that the FAP cluster 1 has has access to three channels which are equally shared by the six incumbent FAPs.
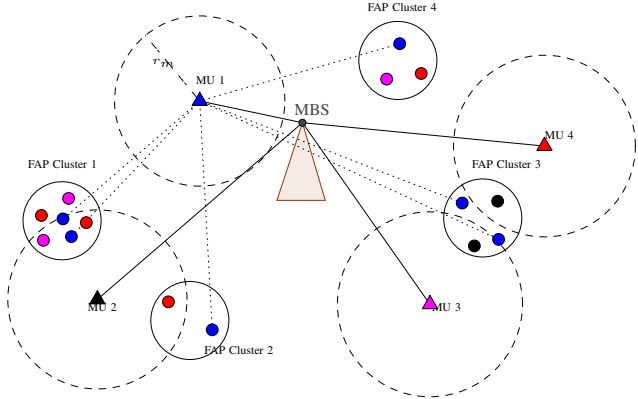
\begin{figure}[h]
  \resizebox{.5\textwidth}{!}{ \tiny{
  \begin{tikzpicture}[line cap=round,line join=round,>=triangle 45,x=1.0cm,y=1.0cm]
\clip(5.,7.) rectangle (15.,14.);
\fill[color=zzttqq,fill=zzttqq,fill opacity=0.1] (9.517613570923457,10.225840653562948) -- (9.897613570923458,11.385840653562948) -- (10.237613570923457,10.225840653562948) -- cycle;
\draw [color=zzttqq] (9.517613570923457,10.225840653562948)-- (9.897613570923458,11.385840653562948);
\draw [color=zzttqq] (9.897613570923458,11.385840653562948)-- (10.237613570923457,10.225840653562948);
\draw [color=zzttqq] (10.237613570923457,10.225840653562948)-- (9.517613570923457,10.225840653562948);
\draw (8.363800988253285,11.710185749698653)-- (9.897613570923456,11.385840653562944);
\draw (6.83605140865709,8.745624065482241)-- (9.897613570923456,11.385840653562944);
\draw (11.801237542344722,8.67287408550147)-- (9.897613570923456,11.385840653562944);
\draw (13.510862071892845,11.037248434876522)-- (9.897613570923456,11.385840653562944);
\draw [dash pattern=on 3pt off 3pt] (8.363800988253285,11.710185749698653) circle (1.2716948292985175cm);
\draw [dash pattern=on 3pt off 3pt] (6.83605140865709,8.745624065482241) circle (1.337739403031422cm);
\draw [dash pattern=on 3pt off 3pt] (11.801237542344722,8.67287408550147) circle (1.3815315056613529cm);
\draw [dash pattern=on 3pt off 3pt] (13.510862071892845,11.037248434876522) circle (1.3572551423077008cm);
\draw(6.30862667934379,9.921044226757534) circle (0.5830703606557937cm);
\draw(12.589304980335328,9.962822574657698) circle (0.5830703606557937cm);
\draw(11.322028427363643,12.316336173033676) circle (0.5830703606557937cm);
\draw(8.216504566784678,8.430949818318306) circle (0.5830703606557937cm);
\draw (8.244356798718123,12.149222781433018) node[anchor=north west] {MU 1};
\draw (13.45272417027208,11.327581939396433) node[anchor=north west] {MU 4};
\draw (11.767664138298741,8.765176601519629) node[anchor=north west] {MU 3};
\draw (6.628927346578392,8.751250485552907) node[anchor=north west] {MU 2};
\draw (5.7376559247081955,10.895872344428058) node[anchor=north west] {FAP Cluster 1};
\draw (10.806762136594937,13.235459826837314) node[anchor=north west] {FAP Cluster 4};
\draw (12.36648712486778,10.895872344428058) node[anchor=north west] {FAP Cluster 3};
\draw (8.36969184241862,8.110649151083706) node[anchor=north west] {FAP Cluster 2};
\draw [dash pattern=on 1pt off 1pt on 3pt off 4pt] (8.363800988253285,11.710185749698653)-- (7.552690937590077,12.689628563025552);
\draw (7.589829348282198,12.399892868834009) node[anchor=north west] {$r_m$};
\draw [dotted] (6.322552795310512,9.948896458690976)-- (8.363800988253285,11.710185749698653);
\draw [dotted] (6.447887839011008,9.684300255323263)-- (8.363800988253285,11.710185749698653);
\draw [dotted] (11.349880659297087,12.567006260434669)-- (8.363800988253285,11.710185749698653);
\draw [dotted] (8.550731349986002,8.291688658651088)-- (8.363800988253285,11.710185749698653);
\draw [dotted] (12.282930429067449,10.185640430125249)-- (8.363800988253285,11.710185749698653);
\draw [dotted] (12.8260489517696,9.6425219074231)-- (8.363800988253285,11.710185749698653);
\begin{scriptsize}
\draw [fill=black,shift={(6.83605140865709,8.745624065482241)}] (0,0) ++(0 pt,3.75pt) -- ++(3.2475952641916446pt,-5.625pt)--++(-6.495190528383289pt,0 pt) -- ++(3.2475952641916446pt,5.625pt);
\draw [fill=ffqqqq,shift={(13.510862071892845,11.037248434876522)}] (0,0) ++(0 pt,3.75pt) -- ++(3.2475952641916446pt,-5.625pt)--++(-6.495190528383289pt,0 pt) -- ++(3.2475952641916446pt,5.625pt);
\draw [fill=qqqqff,shift={(8.363800988253285,11.710185749698653)}] (0,0) ++(0 pt,3.75pt) -- ++(3.2475952641916446pt,-5.625pt)--++(-6.495190528383289pt,0 pt) -- ++(3.2475952641916446pt,5.625pt);
\draw [fill=ffqqff,shift={(11.801237542344722,8.67287408550147)}] (0,0) ++(0 pt,3.75pt) -- ++(3.2475952641916446pt,-5.625pt)--++(-6.495190528383289pt,0 pt) -- ++(3.2475952641916446pt,5.625pt);
\draw [fill=uuuuuu] (9.897613570923456,11.385840653562944) circle (1.5pt);
\draw[color=uuuuuu] (10.040825758425237,11.578252026797424) node {MBS};
\draw [fill=qqqqff] (6.322552795310512,9.948896458690976) circle (2.5pt);
\draw [fill=ffqqff] (6.030104360009354,9.628595791456377) circle (2.5pt);
\draw [fill=ffqqqq] (6.,10.) circle (2.5pt);
\draw [fill=ffqqqq] (6.670705694478557,9.89319199482409) circle (2.5pt);
\draw [fill=ffqqff] (6.4061094911108425,10.255271009958856) circle (2.5pt);
\draw [fill=qqqqff] (6.447887839011008,9.684300255323263) circle (2.5pt);
\draw [fill=black] (12.46396993663483,9.545039095656046) circle (2.5pt);
\draw [fill=black] (12.812122835802876,10.213492662058691) circle (2.5pt);
\draw [fill=qqqqff] (12.282930429067449,10.185640430125249) circle (2.5pt);
\draw [fill=qqqqff] (12.8260489517696,9.6425219074231) circle (2.5pt);
\draw [fill=qqqqff] (11.349880659297087,12.567006260434669) circle (2.5pt);
\draw [fill=ffqqff] (11.154915035762981,12.037813853699241) circle (2.5pt);
\draw [fill=ffqqqq] (11.670181326531688,12.121370549499572) circle (2.5pt);
\draw [fill=qqqqff] (8.550731349986002,8.291688658651088) circle (2.5pt);
\draw [fill=ffqqqq] (7.896203899550076,8.70947213765274) circle (2.5pt);
\draw[color=black] (8.940662597054214,12.608784608334835) node {};
\end{scriptsize}
\end{tikzpicture}}}\normalsize
\caption{Typical example of two-tier HetNet with cognitive clustered FAPs.}
\label{fig:illustration}
\end{figure}


\section{Load-aware Performance Analysis of Two-tier HetNets}
Below we present the load-aware performance analysis of two-tier HetNets with clustered cognitive FAPs under co-channel mode. In next two sub-sections we extend the analysis for orthogonal and partial modes, respectively, followed by a sub-section providing network design insights based on the developed analytical framework.
\subsection{Analysis under Co-Channel Spectrum Sharing}
\label{sec:CoverageAnalysisforCo-ChannelAllocation}
Under co-channel mode, both the tiers have access to the entire channel set \mbox{\small{$\mathcal{N}$}}. 
An MBS accesses a typical channel with probability \mbox{\small{$\zeta_C$}}.
Therefore, thinning the process of MBSs with density \mbox{\small{$\lambda_B$}} by a factor of \mbox{\small{$\zeta_C$}} results in a process of co-channel MBS of density \mbox{\small{$\tilde\lambda_B=\zeta_C\lambda_B$}}. 
Moreover, the cluster of FAPs are restricted to access a channel if there exits no MU (accessing same channel) within the exclusion range.  
This creates holes of radius \mbox{\small{$r_m$}} in the clusters process \mbox{\small{$\Phi_p$}}. 
Note that the point process of parent points of co-channel clusters in \mbox{\small{$\Phi_p$}} is not a PPP but a Poisson hole process. Nonetheless, independent thinning of point process of co-channel clusters (i.e. parent points of co-channel clusters) with probability $p_{r_m}$, given by \eqref{eq:ChannelAccessProb}, yields a good approximation with PPP \cite{Haenggi_2012}. Therefore, effective density of FAPs' clusters per channel becomes \mbox{\small{$\tilde\lambda_F=\lambda_Fp_{r_m}$}}.
To ensure the FAPs transmission, we have assumed that there is one dedicated channel for femto-tier.  Therefore, \mbox{\small{$N$}} is replaced by \mbox{\small{$N-1$}} in the probability mass function (p.m.f) of the number of accessible channels \eqref{eq:nChannelAccessibleProb}. The access of channel, chosen to be dedicated, at FAPs is independent of the incumbent MUs. Therefore, coverage analysis with cognitive FAPs is valid only for the MUs those are using remaining set of channels.

In the following we first derive the coverage probability of an FU and an MU, respectively for a given activity factor. Thereafter, exploiting the activity factor and coverage probability coupling, we present a framework to evaluate the activity factor of MBS for fractional load conditions. {The evaluated activity factor is further plugged into the derived expressions to get the actual load-aware coverage probabilities.}
\subsubsection{Femto User coverage probability Analysis}
\label{sec:FemtoUsercoverage probabilityAnalysis}
The coverage probability of an FU is the probability that its SIR is above the threshold \mbox{\small{$\beta_F$}}. 
The coverage analysis is carried out for an FU situated at the center of a cluster. 
The interfering nodes for an FU are categorized in three sets: 1) co-channel MBSs in \mbox{\small{$\Phi_B$}}, 2) intra-cluster co-channel FAPs in \mbox{\small{$\tilde\phi_c$}}, and 3) inter-cluster co-channel FAPs in \mbox{\small{$\tilde\Phi_F\setminus\tilde\phi_c$}}. The Laplace transform (LT) of interference from these sets of interferes are given in Lemma \ref{lemma1}, \ref{lemma1a}, and \ref{lemma2}, respectively. 
\begin{lemma}
\label{lemma2}
 For a given $\zeta_C$, the LT of macro-tier interference \mbox{\small{$I_{\text{mf}}$}} is:\begingroup\makeatletter\def\f@size{8}\check@mathfonts
 \begin{equation}
  \mathcal{L}_{I_{\text{mf}}}\left(s,\zeta_C\right)=\exp\left(-\pi^2\delta\lambda_B\zeta_C\left(s\chi {P_B}\right)^{\delta}\csc\left[\pi\delta\right]\right)
  \label{eq:Laplace_MacroInterference}
 \end{equation} \endgroup
\end{lemma}
\begin{proof}
 Refer \cite{haenggi2009interference}.
\end{proof}
It can be noted that the LT of macro-tier interference \eqref{eq:Laplace_MacroInterference} differ from \cite{haenggi2009interference} in terms of MBSs density. The fractional activity of an MBS causes in reducing the effective density of co-channel MBS ($\zeta_C\lambda_B$) . This fact results in reduced mean macro-tier interference by factor $\zeta_C$.
\begin{lemma}
\label{lemma1}
 For a given \mbox{\small{$r_m$}}, the approximated LT of the intra-cluster interference \mbox{\small{$I_{\text{if}}$}} can be written as:
 \begingroup\makeatletter\def\f@size{8}\check@mathfonts
 \begin{align}
  \mathcal{L}_{I_{\text{if}}}\left(s,r_m\right)\approx\exp&\left(-\frac{ce^{\pi\lambda_m r_m^2/\mu}}{N R^2}\left(s\chi P_F\right)^{\delta}\mathcal{H}_1\left(s\chi P_F,R,\delta\right)\right)
  \label{eq:Laplace_IntraClusterInterference}\\
  \text{where~~~}&\mathcal{H}_1\left(s\chi P_F,R,\delta\right)=\int_{0}^{\frac{1}{R^{2}}{\left(s\chi P_F\right)^{\delta}}}\frac{1}{1+t^{\frac{1}{\delta}}}dt\nonumber
 \end{align}\endgroup
For \mbox{\small{$\delta=\frac{1}{2}$}}, \mbox{\small{$\mathcal{H}_1\left(s\chi P_F,R,\frac{1}{2}\right)=\arctan\left(\frac{R^2}{(s\chi P_F)^{\frac{1}{2}}}\right)$}}.
\end{lemma}
\begin{proof}
 Refer Appendix \ref{app:lemma1}.
\end{proof}
It can be noted that the inter-cluster and inter-tier interference statistics for an FU is independent of its location in a cluster. Therefore, to justify the consideration of an FU at the center of the cluster in Lemma \ref{lemma1}, we show that the intra-cluster interference is almost independent of an FUs location.
 Let the FU be at \mbox{\small{$y\in\mathcal{B}(0,R)$}}. Thus \eqref{eq:lemma11} can be rewritten as\begingroup\makeatletter\def\f@size{8}\check@mathfonts
 \begin{equation}
  \mathcal{L}_{I_\text{if}}(s,y)=\exp\left(-\frac{ce^{\pi\lambda_mr_m^2/\mu}}{N\pi R^2}\int_{\|x\|\leq R}\frac{dx}{1+\frac{1}{s\chi P_F}\|x-y\|^\alpha}\right)
 \end{equation}\endgroup
 Therefore, the first moment of \mbox{\small{$I_{\text{if}}$}} becomes\begingroup\makeatletter\def\f@size{8}\check@mathfonts
\begin{align}
 \bar I_{\text{if}}(y)&=-\frac{d}{ds}\mathcal{L}_{I_\text{if}}(s,y)\big |_{s=0}\nonumber\\
 &=\frac{ce^{\pi\lambda_mr_m^2/\mu}}{N\pi R^2}\int_{\|x\|\leq R}P_F\|x-y\|^{-\alpha} dx
\end{align}\endgroup
and the second moment of \mbox{\small{$I_{\text{if}}$}} becomes\begingroup\makeatletter\def\f@size{8}\check@mathfonts
\begin{align}
 &\hat I_{\text{if}}(y)=\frac{d^2}{ds^2}\mathcal{L}_{I_\text{if}}(s,y)\big |_{s=0}\nonumber\\
 &=\frac{2ce^{\pi\lambda_mr_m^2/\mu}}{N\pi R^2}\int_{\|x\|\leq R}\left[P_f^2\|x-y\|^{-2\alpha}+P_f\|x-y\|^{-\alpha}\right]dx
\end{align}\endgroup
By plotting above expression it can be realized that these moments are almost independent of location of an FU. 
\begin{lemma}
\label{lemma1a}
  For a given \mbox{\small{$r_m$}}, the approximated LT of the inter-cluster interference \mbox{\small{$I_{\text{of}}$}} can be written as:\begingroup\makeatletter\def\f@size{8}\check@mathfonts
 \begin{equation}
  \mathcal{L}_{I_{\text{of}}}\left(s,r_m\right)\approx\exp\left(-\pi^2\delta\frac{\lambda_F c}{N}\left(s\chi {P_F}\right)^{\delta}\csc\left[\pi\delta\right]\right)
  \label{eq:Laplace_InterClusterInterference}
 \end{equation}\endgroup
\end{lemma}
\begin{proof}
 Refer Appendix \ref{app:lemma1a}.
\end{proof}
It can be seen that the above expression of the LT of interference from clustered FAPs is equivalent to the LT of interference from homogeneous MBSs given in \eqref{eq:Laplace_MacroInterference}. In \cite{RKGanti_2009}, it is shown that the LT of interference realized through the Poisson cluster process is lower bounded by the LT of interference realized through the homogeneous PPP of same density. However, the approximation in \eqref{eq:Laplace_InterClusterInterference} occurs because of applying upper bound on lower bounded expression while taking expectation over Binomial distributed random variable \eqref{eq:nChannelAccessibleProb} which decided the co-channel FAPs in a cluster.
\begin{theorem}
\label{theorem:coverage probability_FU_CoChannel}
 For given $r_m$ and $\zeta_C$, the approximated coverage probability of an FU in two-tier network using co-channel mode is given by 
 \begingroup\makeatletter\def\f@size{8}\check@mathfonts
   \begin{align}
\mathbf{F}_C&\left(\beta_F|r_m,\zeta_C\right)\approx\exp\left\{-\delta\pi^2r_0^2(\beta_F\chi)^{\delta}\left[\frac{\lambda_Fc}{N}+\lambda_B\zeta_C\left(\frac{P_B}{P_F}\right)^{\delta}\right]\csc\left[\pi\delta\right]\right.\nonumber\\
&\left.\vphantom{\left[\frac{\lambda_Fc}{N}+\lambda_B\zeta_C\left(\frac{P_B}{P_F}\right)^{\delta}\right]}-\frac{ce^{\pi\lambda_m r_m^2/\mu}}{N}\frac{r_0^2}{R^2}\left(\beta_F\chi\right)^\delta\mathcal{H}_1\left(r_0^{2/\delta}\chi\beta_F,R,\delta\right) \right\}  \label{eq:FemtoUsercoverage probability}
 \end{align}\endgroup
\end{theorem}
\begin{proof}
 The coverage probability of an FU can be written as\begingroup\makeatletter\def\f@size{8}\check@mathfonts
 \begin{align}
  &\mathbf{F}_C\left(\beta_F|r_m,\zeta_C\right)=\mathbb{P}\left(\frac{h_0r_0^{-\alpha}P_F}{I_{\text{if}}+I_{\text{of}}+I_{\text{mf}}}\geq \beta_F\right)\nonumber\\
  &=\mathbb{P}\left(h_0\geq \frac{\beta_Fr_0^{\alpha}}{P_F}\left(I_{\text{if}}+I_{\text{of}}+I_{\text{mf}}\right)\right)\nonumber\\
  &=\mathbb{E}_{I_{\text{mf}}}\left[\exp\left(-\frac{\beta_Fr_0^{\alpha}}{P_F}I_{\text{mf}}\right)\right]\mathbb{E}_{I_{\text{if}}}\left[\exp\left(-\frac{\beta_Fr_0^{\alpha}}{P_F}I_{\text{if}}\right)\right]\nonumber\\
  &~~\mathbb{E}_{I_{\text{of}}}\left[\exp\left(-\frac{\beta_Fr_0^{\alpha}}{P_F}I_{\text{of}}\right)\right]\nonumber\\
  &=\mathcal{L}_{I_{\text{mf}}}\left(s,\zeta_C\right)\mathcal{L}_{I_{\text{if}}}\left(s,r_m\right)\mathcal{L}_{I_{\text{of}}}\left(s,r_m\right)\big|_{s=\frac{\beta_Fr_0^{\alpha}}{P_F}}
  \label{eq:FemtoUsercoverage probability1}
 \end{align}\endgroup
Substituting \eqref{eq:Laplace_IntraClusterInterference}, \eqref{eq:Laplace_InterClusterInterference} and \eqref{eq:Laplace_MacroInterference} into \eqref{eq:FemtoUsercoverage probability1} completes the proof.
\end{proof}
\subsubsection{Macro User coverage probability Analysis}
\label{sec:MacroUsercoverage probabilityAnalysis}
The coverage probability of an MU is the probability that its SIR is above the threshold \mbox{\small{$\beta_M$}}. There are two sets of interferers for an MU: 1) co-channel FAPs in \mbox{\small{$\tilde\Phi_F$}}, and 2) co-channel MBSs in \mbox{\small{$\Phi_B$}}. The LT of interference from the interferes in \mbox{\small{$\tilde \Phi_F$}} is given in Lemma \ref{lemma3}. Next, we provide the LT of interference from the interferers in \mbox{\small{$ \Phi_B$}} in Lemma \ref{lemma5}.
\begin{lemma}
 \label{lemma3}
 For a given $r_m$, the approximated LT of interference \mbox{\small{$I_{\text{fm}}$}} from co-channel FAPs in PCP is\footnote{$\mathcal{B}^c_{r_m}$ denotes the complement of a ball of radius $r_m$ centered at origin, i.e. $\mathbb{R}^2\setminus\mathcal{B}(0,r_m)$.}:\begingroup\makeatletter\def\f@size{8}\check@mathfonts
 \begin{align}
  &\mathcal{L}_{I_{\text{fm}}}\left(s,r_m\right)\approx\exp\left(-\lambda_F\frac{c}{N}\int_{\mathcal{B}^c_{r_m}}\int_{\mathcal{B}(0,R)}\frac{f(x)dx}{1+\frac{1}{s\chi P_F}\|x-y\|^{\alpha}}dy\right)  \label{eq:Laplace_PCPInterference1}
 \end{align}\endgroup
\end{lemma}
\noindent and for $r_m=0$ (i.e. with non-cognitive FAPs), the approximated LT of \mbox{\small{$I_{\text{fm}}$}} becomes:\begingroup\makeatletter\def\f@size{8}\check@mathfonts
\begin{align}
  \mathcal{L}_{I_{\text{fm}}}\left(s,0\right)\approx&\exp\left(-\pi^2\delta\lambda_F \frac{c}{N}(s\chi P_F)^\delta\csc[\pi\delta]\right) 
  \label{eq:Laplace_PCPInterference_NonCog}
\end{align}\endgroup
\begin{proof}
 Refer Appendix \ref{app:lemma3}.
\end{proof}

The LT of $I_\text{fm}$ with $r_m=0$ (i.e. with non-cognitive FAPs) is equivalent to that is realized through PPP of intensity $\frac{\lambda_Fc}{N}$ \cite{haenggi2009interference}. Further,  it may be noted that the LT of $I_\text{fm}$ with non-cognitive FAPs and LT of $I_\text{of}$ with cognitive FAPs are also equivalent. This implies that the effective density of co-channel nodes is the same in both the cases. It is attributed to the fact that the inclusion of cognition in the clusters of FAPs, the cluster density is thinned by a factor of $\exp(-\pi\lambda_mr_m^2/\mu)$ and the mean number of co-channel FAPs in a cluster are scaled by a factor $\frac{1}{N}\exp(\pi\lambda_mr_m^2/\mu)$. However, in  case on non-cognitive FAPs, cluster density remains unchanged and the mean number of co-channel FAPs in a cluster are scaled by factor $\frac{1}{N}$. Therefore, the effective density of co-channel FAPs becomes $\frac{\lambda_Fc}{N}$ with cognitive as well as non-cognitive FAPs.  
 \begin{lemma}
\label{lemma5}
 For given $\zeta_C$, the LT of interference \mbox{\small{$I_{\text{mm}}(r)$}} from co-channel MBSs outside of \mbox{\small{$\mathcal{B}\left(0,r\right)$}} is:\begingroup\makeatletter\def\f@size{8}\check@mathfonts
 \begin{align}
  \mathcal{L}_{I_{\text{mm}}}\left(s,r,\zeta_C\right)&=\exp\left(-\pi\lambda_B\zeta_C\left(sP_B\right)^\delta\mathcal{H}_2\left(\frac{sP_B}{r^{2/\delta}},\delta\right)\right)  \label{eq:Laplace_MacroInterferenceBeyond_d0}\\
  \text{where~~}\mathcal{H}_2\left(\frac{sP_B}{r^{2/\delta}},\delta\right)&=\int_{\frac{r^2}{(sP_B)^\delta}}^\infty\frac{dt}{1+t^{\frac{1}{\delta}}}.\nonumber
 \end{align}\endgroup
\noindent For \mbox{\small{$\delta=\frac{1}{2}$}}, \mbox{\small{$\mathcal{H}_2\left(\frac{sP_B}{r^{4}},\frac{1}{2}\right)=\arctan\left(\frac{(sP_B)^{\frac{1}{2}}}{r^2}\right)$}}. 
\end{lemma}
\begin{proof}
 Refer \cite{HSJo_2012}.
\end{proof}
\begin{theorem}
\label{theorem:coverage probability_MU_CoChannel}
  For given $r_m$ and $\zeta_C$, the approximated coverage probability of an MU in two-tier network using co-channel mode is given by \eqref{eq:MacroUsercoverage probability} (given at the top of next page). 
  \begin{figure*}[!t]
 \begingroup\makeatletter\def\f@size{8}\check@mathfonts
   \begin{align}
   &\mathbf{M}_C\left(\beta_M|r_m,\zeta_C\right)\approx\int\limits_{0}^{\infty}2\pi\lambda_Br\exp\left\{-\pi\lambda_Br^2\left(1+\zeta_C\beta_M^\delta\mathcal{H}_2\left(\beta_M,\delta\right)\right)-\frac{\lambda_F}{\pi R^2}\frac{c}{N}\int\limits_{\mathbb{R}^2\setminus\mathcal{B}(0,r_m)}\int\limits_{\mathcal{B}(0,R)}\frac{dx}{1+\frac{P_M}{\beta_M\chi P_Fr^\alpha}\|x-y\|^{\alpha}}dy\right\}dr   
  \label{eq:MacroUsercoverage probability}
 \end{align}\endgroup 
 \hrule
 \end{figure*} 
  \newline For $r_m=0$ (i.e. with non-cognitive FAPs) the approximated coverage probability of an MU becomes:\begingroup\makeatletter\def\f@size{8}\check@mathfonts
  \begin{align}
  \begin{split}
   \mathbf{M}_C\left(\beta_M|r_m,\zeta_C\right)\approx\left[1+\zeta_c\beta_M^\delta\mathcal{H}_2(\beta_M,\delta)\vphantom{\pi\delta\frac{\lambda_Fc}{\lambda_BN}(\beta_M\chi)^\delta\csc[\delta\pi]}\right.\\
   \left.+\pi\delta\frac{\lambda_Fc}{\lambda_BN}\left(\beta_M\chi{P_F}/{P_B}\right)^\delta\csc[\delta\pi]\right]^{-1}
  \end{split}   
   \label{eq:MacroUsercoverage probability_NonCog}
  \end{align}\endgroup
\end{theorem}
\begin{proof}
 The coverage probability of an MU when serving BS is at a distance \mbox{\small{$r$}} can be written as\begingroup\makeatletter\def\f@size{8}\check@mathfonts
 \begin{align}
  &\mathbf{M}_C\left(\beta_M|r,r_m,\zeta_C\right)=\mathbb{P}\left(\frac{hd^{-\alpha}P_B}{I_{\text{mm}}(r)+I_{\text{fm}}}\geq \beta_M\right)\nonumber\\
  &=\mathcal{L}_{I_{\text{mm}}}\left(s,r,\zeta_C\right)\times\mathcal{L}_{I_{\text{fm}}}\left(s,r_m\right)\big|_{s=\frac{\beta_Mr^{\alpha}}{P_B}}
  \label{eq:MacroUsercoverage probability1}
 \end{align}\endgroup
Note that the probability density function of distance of an MBS from origin is given by \mbox{\small{$f_R(r)=2\pi\lambda_Br\exp(-\pi\lambda_Br^2)$}} in \cite{HSJo_2012}. Therefore, the coverage probability of an MU becomes
\begingroup\makeatletter\def\f@size{8}\check@mathfonts
\begin{align}
 &\mathbf{M}_C\left(\beta_M|r_m,\zeta_C\right)=\int_0^{\infty}\mathbf{M}_C\left(\beta_M|r,r_m,\zeta_C\right)f_R(r)dr\nonumber\\
 &=\int_0^{\infty}\mathcal{L}_{I_{\text{mm}}}\left(\frac{\beta_Mr^{\alpha}}{P_B},r,\zeta_C\right)\times\mathcal{L}_{I_{\text{fm}}}\left(\frac{\beta_Mr^{\alpha}}{P_B},r_m\right)f_R(r)dr
 \label{eq:MacroUsercoverage probability_averaging}
\end{align}\endgroup
Substituting \eqref{eq:Laplace_PCPInterference1} and \eqref{eq:Laplace_MacroInterferenceBeyond_d0} into \eqref{eq:MacroUsercoverage probability_averaging} yields \eqref{eq:MacroUsercoverage probability}. Further, substituting \eqref{eq:Laplace_PCPInterference_NonCog} and \eqref{eq:Laplace_MacroInterferenceBeyond_d0} into \eqref{eq:MacroUsercoverage probability_averaging} yields \eqref{eq:MacroUsercoverage probability_NonCog}. This completes the proof.
\end{proof}
\noindent{Through simulation validation, it is proven that the derived approximated expressions for coverage probability of an MU and an FU are tighter.}
\subsubsection{Modeling Activity Factor of a Macro BS}
\label{sec:ModelingActivityFactorofaMacroBS}
The underlying MU services are considered to be admitted into the network under the assumptions of real-time traffic. Therefore, the scenario of multiple channels and fractional load conditions decides the bandwidth occupancy and thus the interference realization. The interference dynamics and bandwidth occupancy are related through the activity factor of an MBS. In a random cellular network, the activity factor can be interpreted as follows.
  \begin{enumerate}
   \item Probability that a BS randomly chooses a channel.
   \item Average fraction of BSs in the network those are co-channel at a given instant.  
   \item Fraction of total spectrum/channels  accessed by a BS. 
  \end{enumerate}  
   There is implicit coupling between the activity factor and the interference (or equivalently coverage probability). For example, transmission from a BS, say B1, generates interference to its neighboring BS, say B2, which forces B2 to transmit for a longer time which again interferes back to B1 and make B1 to transmit for even longer time. The coupling between interference and activity factor can be mathematically modeled using the interdependence between the four entities as described below: 1) \emph{interference} (thus the coverage probability) determines the \emph{achievable rate}, 2) \emph{achievable rate} along with traffic load determines the \emph{bandwidth occupancy}, 3) \emph{bandwidth occupancy} determines the \emph{BS activity factor}, and 4) \emph{BS activity factor} decides the \emph{interference}.    

   This makes the load dependent fractional activity of BS as a pivotal element in the co-channel interference modeling. In this section, we present a model to evaluate the activity factor \mbox{\small{$\zeta_C$}}. The activity factor is defined as the probability that an MBS uses a typical channel. The activity factor depends on the traffic intensity and the resource units required per service. In cellular environment, the resource requirement basically depends on SIR statistics experienced by the service which is further location dependent. Therefore, a cell can be modeled using multi-dimensional queue for the evaluation of number of busy servers. However, we modeled a cell using single dimensional queue with the consideration of the average number of resource units required by the service to avoid location dependency and to keep the model tractable. Therefore, using \eqref{eq:MacroUsercoverage probability}, the mean number of channels required by an MU service for a given \mbox{\small{$r_m$}} and \mbox{\small{$\zeta_C$}} can be calculated as follows
\begingroup\makeatletter\def\f@size{8}\check@mathfonts
\begin{align}
 \bar N\left(\zeta_C,r_m\right)&=\sum\limits_{i=1}^{T}\frac{R_{\text{th}}}{B\log_2\left(1+\Gamma_i\right)}\text{Prob}(\Gamma_i),
 \label{eq:AVgNumberOfChannelPerMacroService}\\
 \text{where~~}\text{Prob}(\Gamma_i)&=\frac{\mathbf{M}_C\left(\Gamma_{i+1}|r_m,\zeta_C\right)-\mathbf{M}_C\left(\Gamma_{i}|r_m,\zeta_C\right)}{1-\mathbf{M}_C\left(\Gamma_{1}|r_m,\zeta_C\right)},\nonumber
\end{align}\endgroup
\mbox{\small{$\Gamma_i$}} is the SINR threshold of the \mbox{\small{$i$th}} MCS,  and \mbox{\small{$\Gamma_{T+1}=\infty$}}. Note that the probability of each MCS is normalized by the outage probability (i.e. \mbox{\small{$1-\mathbf{M}_C\left(\Gamma_{1}|r_m,\zeta_C\right)$}}).
For a given cell area $a$ and \mbox{\small{$\bar N(\zeta_C,r_m)$}}, the cell can be modeled using \mbox{\small{$M/M/N_{\text{e}}/N_{\text{e}}$}} queue with service arrival rate \mbox{\small{$a\lambda_{M}$}} and departure rate $\mu$, where \mbox{\small{$N_{\text{e}}=\left\lfloor\frac{N}{\bar N(\zeta_C,r_m)}\right\rfloor$}}. Therefore, the probability that $n$ servers/channels are occupied is \mbox{\small{$\mathbb{P}\left(n|a\right)=\frac{\left({a\lambda_{M}}/{\mu}\right)^n}{n!}/\sum_{k=0}^{N_{\text{e}}}\frac{\left({a\lambda_{M}}/{\mu}\right)^i}{i!}$}} \cite{Bertsekas:1992}. The probability that the MBS uses a channel for a given $a$ is \cite{Wei_Bao_2014_NearOptimal}\begingroup\makeatletter\def\f@size{8}\check@mathfonts
\begin{align}
\label{eq:Channel_Access_Probability_conditional}
\zeta_C\left(a\right)&=\frac{1}{N_{\text{e}}}\sum\limits_{i=0}^{N_{\text{e}}}n\mathbb{P}\left(n|a\right)\approx\begin{cases}
                                                                                         &\frac{a\lambda_{M}}{\mu N_{\text{e}}} \mspace{17mu} \text{if}~ \frac{a\lambda_{M}}{\mu N_{\text{e}}}<1\\
                                                                                         &1 \mspace{40mu} \text{if}~ \frac{a\lambda_{M}}{\mu N_{\text{e}}}\geq 1
                                                                                        \end{cases}
\end{align}\endgroup
The approximation is tighter for larger values of \mbox{\small{$N_{\text{e}}$}}. For larger \mbox{\small{$N$}}, \mbox{\small{$N_{\text{e}}\approx \frac{N}{\bar N(\zeta_C,r_m)}$}}. 
The cell area density function is derived in \cite{Singh_2013} as follows
\begingroup\makeatletter\def\f@size{8}\check@mathfonts
\begin{equation}
\label{eq:cell_size_PDF}
 f_A\left(a\right)=\frac{\left({3.5\lambda_B}\right)^{3.5}}{\Gamma\left(3.5\right)}a^{2.5}\exp\left(-3.5a\lambda_B\right)
\end{equation}\endgroup
Therefore, using \eqref{eq:Channel_Access_Probability_conditional} and \eqref{eq:cell_size_PDF} we have derived the activity factor of an MBS in PPP {{$\Phi_B$}} in \cite{Praful_BlockingProb} as 
\begingroup\makeatletter\def\f@size{8}\check@mathfonts
 \begin{equation}
  \zeta_C=\frac{1}{3.5\varphi\Gamma(3.5)}\gamma\left(4.5,3.5\varphi\right)+\frac{1}{\Gamma(3.5)}\Gamma\left(3.5,3.5\varphi\right)
  \label{eq:ActivityFactorApp}
 \end{equation}\endgroup
where \mbox{\small{$\frac{1}{\varphi}=\frac{\lambda_M}{\lambda_B\mu N_{\text{e}}}=\frac{\lambda_M\bar N\left(\zeta_C,r_m\right)}{\lambda_B\mu N}$}}, and \mbox{\small{$\gamma(\cdot,\cdot)$}} and \mbox{\small{$\Gamma(\cdot,\cdot)$}} are lower and upper incomplete gamma functions.

After substitution of \eqref{eq:AVgNumberOfChannelPerMacroService} into \eqref{eq:ActivityFactorApp}, the resultant expression becomes recursive in nature. Therefore, 
we solve it using the bisection method to evaluate the activity factor for a given traffic intensity (\mbox{\small{$\frac{\lambda_B\mu}{\lambda_M}$}}) and \mbox{\small{$r_m$}}. {The determined value of $\zeta_C$ is further substituted in \eqref{eq:FemtoUsercoverage probability} and \eqref{eq:MacroUsercoverage probability} to evaluate the actual coverage probability of an FU and an MU, respectively, for a given macro-tier fractional-load condition.}
\subsection{Analysis under Orthogonal Spectrum Sharing}
\label{sec:CoverageAnalysisforOrthogonalSpectrumAllocation}

Let \mbox{\small{$N_F$}} and \mbox{\small{$N-N_F$}} be the number of channels allocated to the macro-tier and femto-tier, respectively. As the femto cluster has access to \mbox{\small{$N_F$}} channels, the effective mean number of femto  cells per cluster becomes \mbox{\small{$\tilde c=\frac{c}{N_F}$}} assuming the channels are equally shared by the femto cells within the cluster. 
\begin{corollary}
\label{theorem:coverage probability_FU_Orthogonal}
 The approximated coverage probability of an FU in two-tier network using orthogonal mode is \begingroup\makeatletter\def\f@size{8}\check@mathfonts 
 \begin{equation}
 \begin{split}
  \mathbf{F}_C\left(\beta_F|r_m,N_F\right)\approx\exp\left\{-\delta\pi^2\frac{\lambda_Fc}{N_F}r_0^2(\beta_F\chi)^{\delta}\csc\left[\pi\delta\right]\right.\\
  \left.-\frac{c}{N_F}\frac{r_0^2}{R^2}\left(\beta_F\chi\right)^\delta\mathcal{H}_1\left(r_0^{2/\delta}\chi\beta_F,R,\delta\right) \right\}
 \end{split}  
 \end{equation}\endgroup
\end{corollary}
\begin{proof}
 Replacing \mbox{\small{$N$}} by \mbox{\small{$N_F$}}, \mbox{\small{$r_m$}} by 0 (i.e. non-cognitive access), and \mbox{\small{$\lambda_B\zeta_C$}} by 0 (i.e. no macro-tier interference) in the expression of an FU coverage probability in co-channel mode (\mbox{\small{$\mathbf{F}_C(\beta_F|r_m,\zeta_C)$}}) given in Theorem \ref{theorem:coverage probability_FU_CoChannel}, completes the proof.
\end{proof}
Let $\zeta_O$ represent the activity factor of an MBS under orthogonal mode.
\begin{corollary}
\label{theorem:coverage probability_MU_Orthogonal}
 For given \mbox{\small{$\zeta_O$}}, the coverage probability of an MU in two-tier network using orthogonal mode is
 \begingroup\makeatletter\def\f@size{8}\check@mathfonts
   \begin{equation}
   \mathbf{M}_O\left(\beta_M|\zeta_O\right)=\left[1+\zeta_O\beta_M^\delta\mathcal{H}_2\left(\beta_M,\delta\right)\right]^{-1}   
  \label{eq:MacroUsercoverage probability_ortho}
 \end{equation}\endgroup
\end{corollary}
\begin{proof}
 Replacing \mbox{\small{$\lambda_F$}} and \mbox{\small{$c$}} by 0 (i.e. no femto-tier interference), and further solving the integral in the expression of an MU coverage probability in co-channel mode, i.e. \mbox{\small{$\mathbf{M}_C(\beta_M|r_m,\zeta_O)$}} given in Theorem \ref{theorem:coverage probability_MU_CoChannel}, completes the proof.
\end{proof}
The activity factor evaluation of an MBS in orthogonal mode follows the same process as of the co-channel mode whereas the \mbox{\small{$\mathbf{M}_C(\Gamma_i|r_m,\zeta_C)$}} is replaced by \mbox{\small{$\mathbf{M}_O(\Gamma_i|\zeta_O)$}} in \eqref{eq:AVgNumberOfChannelPerMacroService} and \mbox{\small{$N$}} is replaced by \mbox{\small{$N-N_F$}}.
\subsection{Analysis under Partial Spectrum Sharing}
\label{sec:CoverageAnalysisforPartialSpectrumAllocation}
In partial mode, \mbox{\small{$N_F$}} number of channels are allocated for the co-channel deployment of macro-tier and femto-tier, and rest number of channels are reserved for macro-tier only. This helps in reducing the femto-tier interference to an MU by factor of $\frac{N-N_F}{N}$ compared to co-channel mode.
Therefore, it is obvious that the partial mode always yields better coverage for an MU as compared to that of co-channel mode. 
\begin{corollary}
\label{theorem:coverage probability_FU_Partial}
 For given $r_m$ and $\zeta_P$, the approximated coverage prob. of an FU in two-tier network using partial mode is given by \eqref{eq:FemtoUsercoverage probability} with $N=N_F$ and $\zeta_C=\zeta_P$.
\end{corollary}
\begin{proof}
 The coverage of an FU in partial mode is same as it is in co-channel mode except $N$ is replaced by $N_F$.
\end{proof}
\begin{corollary}
\label{theorem:coverage probability_MU_Partial}
 For given $r_m$ and $\zeta_P$, the approximated coverage prob. of an MU in two-tier network using partial mode is 
 \begingroup\makeatletter\def\f@size{8}\check@mathfonts
 \begin{equation}
   \mathbf{M}_P\left(\beta_M|r_m,\zeta_P\right)=p_c \mathbf{M}_C\left(\beta_M|r_m,\zeta_P\right)+(1-p_c)\mathbf{M}_O\left(\beta_M|\zeta_P\right)
  \label{eq:MacroUsercoverage probability_ortho}
 \end{equation}\endgroup
 where \mbox{\small{$p_c=\frac{N_F}{N}$}}.
\end{corollary}
\begin{proof}
 It is assumed that an MU uniformly accesses a channel from pool of \mbox{\small{$N$}} channels. Therefore, an MU experiences coverage probability with co-channel mode, i.e. \mbox{\small{$\mathbf{M}_C(\beta_M|r_m,\zeta_P)$}} given in Theorem \ref{theorem:coverage probability_MU_CoChannel},  with probability \mbox{\small{$p_c$ (=$\frac{N_F}{N}$)}} and an MU experiences coverage probability with orthogonal mode, i.e. \mbox{\small{$\mathbf{M}_O(\beta_M|\zeta_P)$}} given in Corollary \ref{theorem:coverage probability_MU_Orthogonal},  with probability \mbox{\small{$1-p_c$}}. This completes the proof. 
\end{proof}
The activity factor evaluation of an MBS in partial mode follows the same process as of the co-channel mode whereas the \mbox{\small{$\mathbf{M}_C(\Gamma_i|r_m,\zeta_C)$}} is replaced by \mbox{\small{$\mathbf{M}_P(\Gamma_i|r_m,\zeta_P)$}} in \eqref{eq:AVgNumberOfChannelPerMacroService}.
\subsection{Design Insights}
\label{sec:FurtherAnalysis}
The activity factor of an MBS can be approximate to be a linear function of service arrival rate as follows
  \begin{equation}
   \zeta_C=\begin{cases}
		&\frac{\lambda_M \bar N}{\lambda_B \mu N}  \text{~~~if~~}\frac{\lambda_M \bar N}{\lambda_B \mu N}  \leq 1\\
		&1 \text{~~~~~~~~~otherwise}
           \end{cases}   
  \end{equation}
From Fig. \ref{fig:LinearApprx}, it can be seen that the linear approximation is tighter till activity of 0.6 wherein the coverage is seen to be sensitive to the activity factor. The curvature beyond 0.6 is attributed to the fact of raising blocking probability with increase in load. Therefore, it is essential to design network such that the activity lies within the linear region.  It can be noted that the activity factor is equal to the traffic intensity per MBS per channel (i.e. $\frac{\lambda_M}{\lambda_B \mu N}$) times mean number of channels required by a service (i.e. $\bar N$). For case of fixed bandwidth allocation, above approximated model can simplify the load-aware performance evaluation to a great extent.
\begin{figure}
\begin{center}
\centering
\includegraphics[width=.4\textwidth]{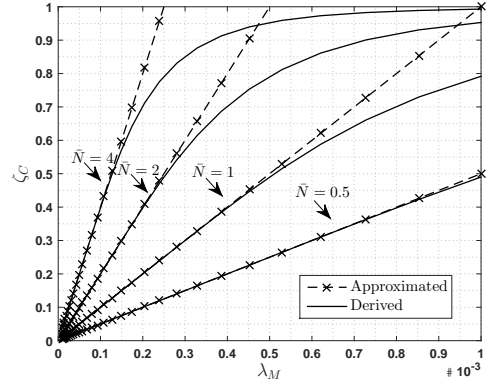} \\
\caption{Liner approximation of activity factor.}
\label{fig:LinearApprx}
\end{center} 
\end{figure}

It may be noted that under co-channel mode, the coverage probability of MUs degrade with increase in density of FAPs since inter-tier interference is proportional to density of FAPs. On the other hand, the coverage probability in orthogonal mode depends on the spectrum partitioning.
The spectrum partitioning increases the activity factor of an MBS results in increase intra-tier interference. Therefore, it is interesting to understand the interdependency between density of FAPs and spectrum partitioning in terms of regions  of co-channel mode and orthogonal mode is effective.
It is clear that till certain density of FAPs the co-channel mode provide better coverage and beyond that point the orthogonal mode dominates the co-channel mode. 
By equating the coverage probabilities under co-channel and orthogonal modes and further solving the expression for the density of FAPs we can obtain the critical point beyond which the orthogonal mode dominates the co-channel mode. 
Therefore, to gain the insight into the system we analyze the developed analytical framework under three special cases: 1) no cognition at FAPs and  mean number of channel per MU is one [\mbox{\small{$r_m=0$}} and \mbox{\small{$\bar N=1$}}], 2) no cognition at FAPs and actual mean number of channel per MU [\mbox{\small{$r_m=0$}} and \mbox{\small{$\bar N$}} as per \eqref{eq:AVgNumberOfChannelPerMacroService}], and 3) with cognitive FAPs [\mbox{\small{$r_m\neq 0$}}].

\begin{table*}[!t]
\begin{center}
\caption{Effective densities of co-channel MBSs and FAPs for co-channel, partial, and orthogonal spectrum access methods}
\label{table:EffectiveDensities} 
\small{
{\renewcommand{\arraystretch}{1.5}
\begin{tabular}{|l|c|c|c|c| }
\hline
&\multicolumn{2}{|c|}{{MUs}} & \multicolumn{2}{|c|}{{FUs}}\\ \hline
&  MBSs & FAPs& MBSs& FAPs\\
\hline
Co-channel&$\zeta_\text{C}\lambda_B$& $p_{r_m}\lambda_F\in\mathbb{R}^2\setminus\mathcal{B}(0,r_m)$ and $\tilde c\approx\frac{c}{p_{r_m}N}$ & $\zeta_\text{C}\lambda_B$ &$p_{r_m}\lambda_F$ and $\tilde c\approx\frac{c}{p_{r_m}N}$\\ \hline

&&For $N_F$ channels: &&\\ 
Partial&$\zeta_\text{S}\lambda_B$& $p_{r_m}\lambda_F\in\mathbb{R}^2\setminus\mathcal{B}(0,r_m)$ and $\tilde c\approx\frac{c}{p_{r_m}N_F}$ &$\zeta_\text{S}\lambda_B$&$p_{r_m}\lambda_F$ and $\tilde c\approx\frac{c}{p_{r_m}N_F}$\\
&&For $N-N_F$ channels: 0   &&\\ \hline
Orthogonal&$\zeta_\text{O}\lambda_B$&0&0&$\lambda_F$ and $\tilde c\approx\frac{c}{N_\text{F}}$\\ 
\hline
\end{tabular}}\quad}
\end{center}
\end{table*}

\textit{Case 1} [\mbox{\small{$r_m=0$}} and \mbox{\small{$\bar N=1$}}]: The fact of \mbox{\small{$r_m=0$}} allows a cluster of FAPs to access each channel. This makes the effective density of FAPs cluster equal to \mbox{\small{$\lambda_F$}} with \mbox{\small{$\frac{c}{N}$}} mean number of co-channel FAPs in $\mathbb{R}^2$. The fact of \mbox{\small{$\bar N=1$}} makes the activity factor and coverage probability uncoupled. The approximated activity factor reduces to \mbox{\small{$\zeta_C=\frac{\lambda_M}{\lambda_B \mu N}$}} for co-channel mode and  \mbox{\small{$\zeta_O=\frac{\lambda_M}{\lambda_B \mu (N-N_F)}$}} for orthogonal mode.  Therefore, by substituting activity factors in the coverage probabilities and equating them we get\begingroup\makeatletter\def\f@size{8}\check@mathfonts
\begin{align}
 \lambda_F' c'&=\frac{\lambda_M}{\mu}\frac{N_F}{N-N_F}\left(\frac{P_M}{\chi P_F}\right)^{\delta}\frac{\mathcal{H}_2(\beta_M,\delta)}{\delta\pi\csc[\delta\pi]}
\end{align}\endgroup
for $\frac{\lambda_M }{\lambda_B \mu }\leq N-N_F$.
The above expression state that till the density of FAPs is below critical point (i.e. \mbox{\small{$\lambda_F'c'$}}) the co-channel mode dominate the orthogonal mode. 
 An interesting remark can be observed that the critical point is proportional to the density of MUs and the ratio of channels available to FAPs and MBSs in orthogonal mode.   
  \begin{figure*}[!t]
\vspace{-.35cm}
\centering
\subfigure[]{
    \centering
    \includegraphics[width=.4\textwidth]{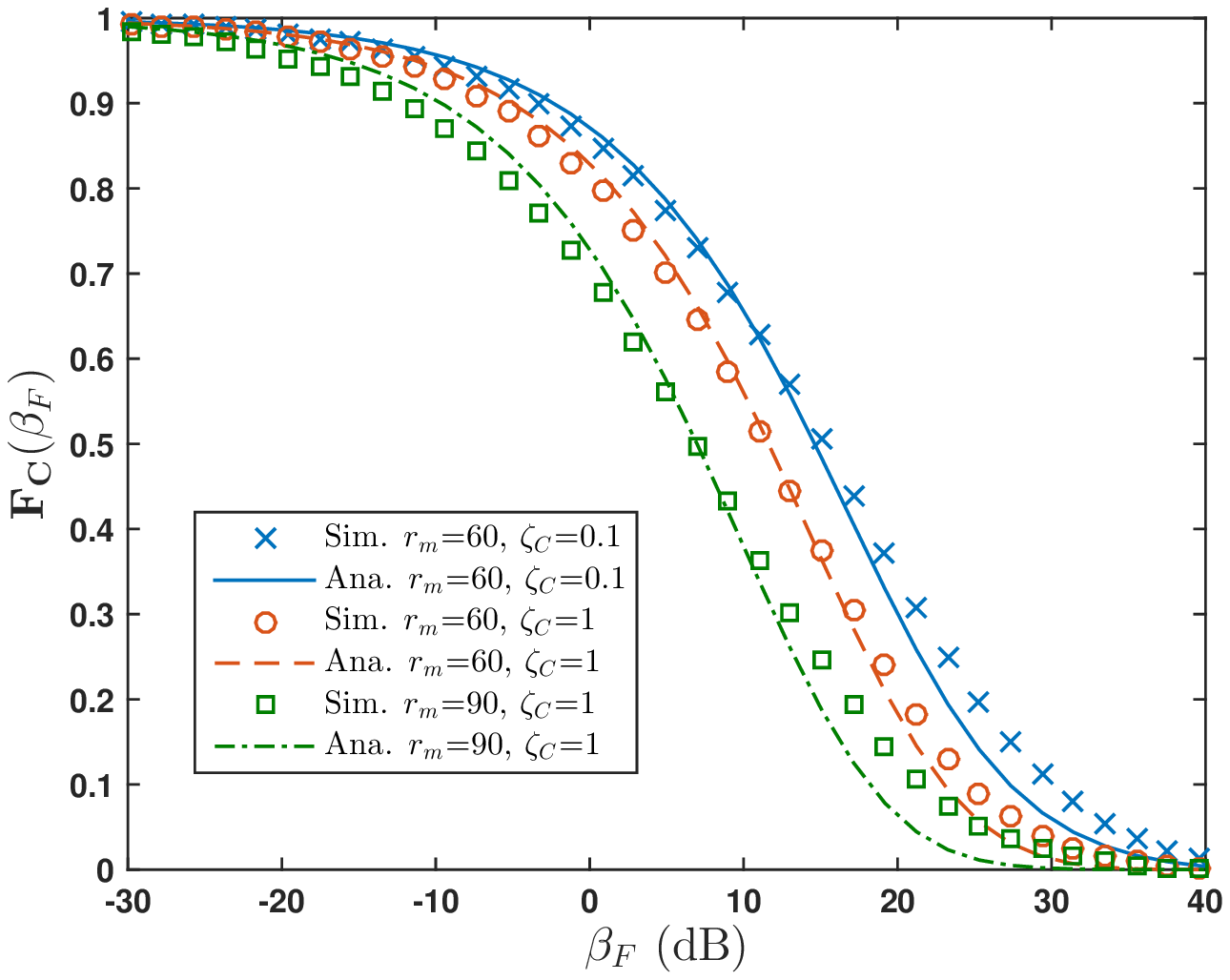} 
    \label{fig:coverage probabilityfemtoUser}
}
\subfigure[]{
    \centering
    \includegraphics[width=.4\textwidth]{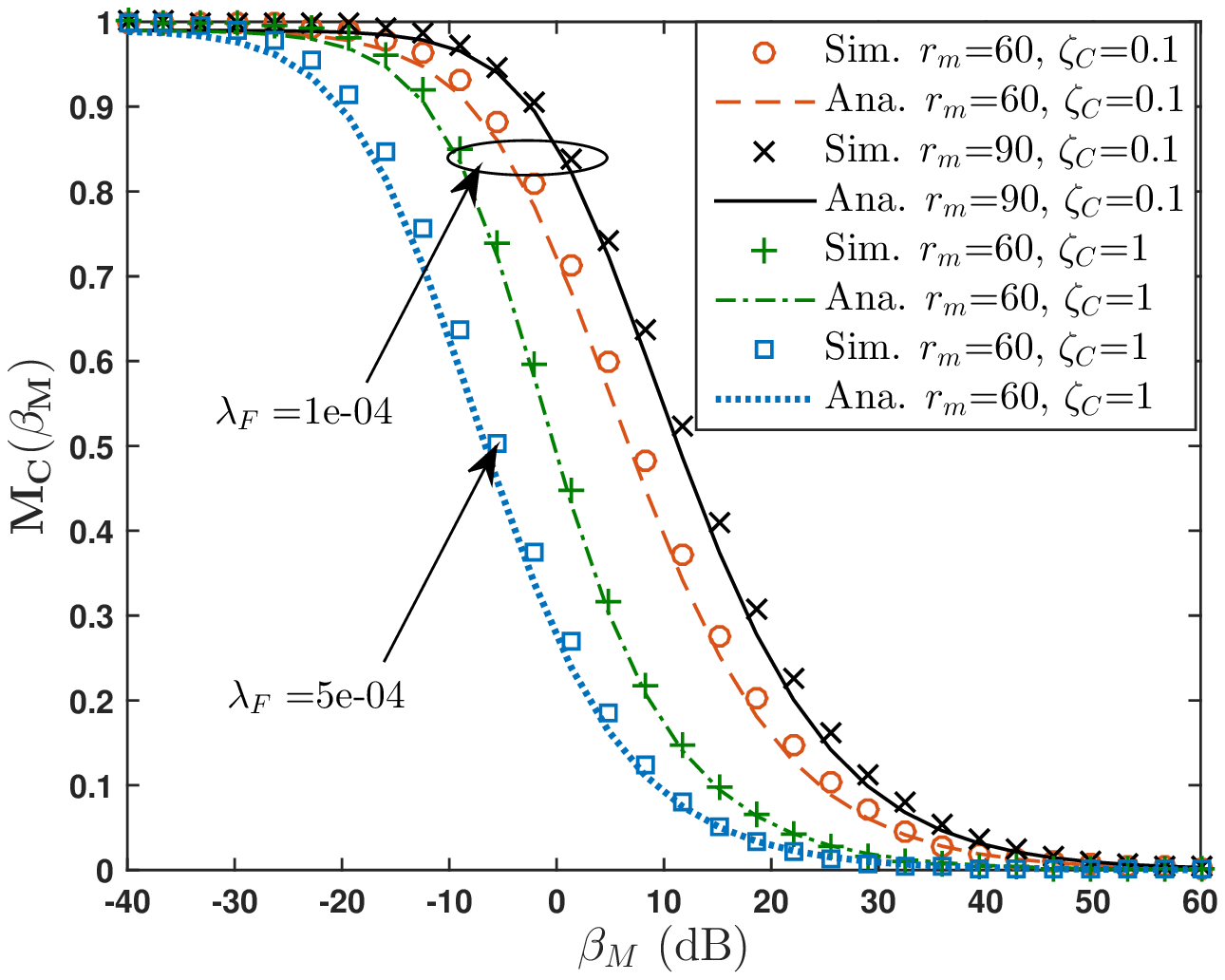} 
    \label{fig:coverage probabilityMemtoUser} 
}\vspace{-.35cm}
\caption{Validation of the coverage probability of \ref{fig:coverage probabilityfemtoUser} FU for \mbox{\small{$\lambda_F=1\text{e}-4$}}, and \ref{fig:coverage probabilityMemtoUser} MU.} 
\label{fig:SimulationValidation}\vspace{-.35cm}
\end{figure*}
 
\textit{Case 2} [\mbox{\small{$r_m=0$}} and \mbox{\small{$\bar N$}} as per \eqref{eq:AVgNumberOfChannelPerMacroService}]: With linear approximation of activity factor and evaluation of mean number of channels \mbox{\small{$\bar N$}} as per \eqref{eq:AVgNumberOfChannelPerMacroService}, the activity factor of an MBS can be rewritten as\begingroup\makeatletter\def\f@size{8}\check@mathfonts
\begin{align}
 \zeta_C&=\frac{\lambda_M}{\lambda_B\mu N}\sum\limits_{i=1}^{T}\frac{n_i-n_{i-1}}{1+\zeta_C \Gamma_i^\delta\mathcal{H}_2(\Gamma_i,\delta)+\pi\delta\frac{\lambda_Fc}{\lambda_BN}\left(\Gamma_i\chi\frac{P_F}{P_B}\right)^\delta\csc[\delta\pi]}\\
 \zeta_O&=\frac{\lambda_M}{\lambda_B\mu (N-N_F)}\sum\limits_{i=1}^{T}\frac{n_i-n_{i-1}}{1+\zeta_O \Gamma_i^\delta\mathcal{H}_2(\Gamma_i,\delta)}
\end{align}\endgroup
where \mbox{\small{$n_0=0$}} and \mbox{\small{$n_i=\frac{R_\text{th}}{B\log_2(1+\Gamma_i)}$}}. Note that \mbox{\small{$n_i<n_{i-1}$}}.
It is difficult to solve above expressions of activity factors  for large value of $T$. Therefor, to provide the insight into the load-aware network performance we considered the nature of activity factor as a function of FAPs density (i.e. \mbox{\small{$\zeta_C=f(\lambda_Fc)$}}) and as a function of \mbox{\small{$N_F$}} (i.e. \mbox{\small{$\zeta_O=g(N_F)$}}). It can be seen that \mbox{\small{$f(\lambda_Fc)$}} is a monotonically increasing function of \mbox{\small{$\lambda_Fc$}} and \mbox{\small{$g(N_F)$}} is also a monotonically increasing function of \mbox{\small{$N_F$}}.  
  
Representing the activity factors as \mbox{\small{$f(\lambda_Fc)$}} and \mbox{\small{$g(N_F)$}} under co-channel mode and orthogonal mode respectively, the solution for the critical point can be determined as follows.\begingroup\makeatletter\def\f@size{8}\check@mathfonts
  \begin{align}
   \lambda_F'c'={\lambda_BN}\left[g(N_F)-f(\lambda_F'c')\right]\left(\frac{P_M}{\chi P_F}\right)^{\delta}&\frac{\mathcal{H}_2(\beta_M,\delta)}{\delta\pi\csc[\delta\pi]}\label{eq:CriticalPoint1}
  \end{align}\endgroup
  for $\frac{\lambda_M }{\lambda_B \mu }\leq N-N_F$.
From above expression is can be seen that the critical point is dependent on the difference of activity factors under orthogonal mode and co-channel mode. In case the difference is negative, then the orthogonal mode is always better. Consider the extreme case of \mbox{\small{$N_F=0$}}, the \mbox{\small{$g(N_F)$}} is obtained to be least as an MBS access full spectrum and inter-tier interference is zero; however at the same time \mbox{\small{$f(\lambda_Fc)$}} is greater than \mbox{\small{$g(N_F)$}} for any value of \mbox{\small{$\lambda_Fc$}} due to inclusion of extra inter-tier interference. This makes the difference to be negative and orthogonal mode becomes the choice of deployment. Therefore, \mbox{\small{$N_F$}}  should be chosen beforehand such that the performance of FAPs is ensured.

\textit{Case 3} [\mbox{\small{$r_m\neq0$}}]: As we do not have the closed form expression for the coverage probability in case of \mbox{\small{$r_m>0$}}, i.e. clustered cognitive FAPs, it is difficult to find the critical point analytically. Nevertheless, to get a meaningful insight of introduction of the cognition in FAPs, here we provide the analysis of impact of \mbox{\small{$r_m$}} on the critical point in comparison with the non-cognitive case. By inclusion of cognition in the clusters of FAPs for the channel access, the inter-tier interference at MUs can be controlled. The inter-tier interference  reduces with increase in the range of exclusion region. Thus, the \mbox{\small{$f(\lambda_Fc)$}} becomes monotonically decreasing function of \mbox{\small{$r_m$}}. This fact helps to increase the value of critical point i.e. the difference of \mbox{\small{$g(N_F)-f(\lambda_F c)>0$}} in \eqref{eq:CriticalPoint1} can be maintained for larger values of \mbox{\small{$\lambda_Fc$}}. 
This implies the co-channel mode of deployment with cognitive FAPs  can provide better the coverage over larger region of values of density of FAPs. 

The analytical evaluation of critical point is difficult as the function $f(\lambda_Fc)$ and $g(N_F)$ are unknown.  However, using the developed analytical framework evaluation of the  critical point via numerical method is always possible.

\subsection{Effective Densities}
Table \ref{table:EffectiveDensities} summaries the impact of co-channel, partial, and orthogonal spectrum sharing modes on the point processes of co-channel MBSs and FAPs. Table \ref{table:EffectiveDensities} depicts the effective densities of interfering co-channel MBSs and FAPs for an MU and an FU. The density of intra-cluster co-channel FAPs is indicated by $\tilde c=\frac{c}{N}\exp(\pi\lambda_mr_m^2/\mu)$. 
From table it is clear that the characterization of interference from co-channel nodes is dependent on the spectrum sharing method. 

\section{Numerical Results and Discussion}
\label{sec:NumericalResultsAndDiscussion}
{In this section, we first validate the derived  expressions for coverage probability of an MU and an FU using the extensive simulations. Next, we present a numerical analysis for activity factor, coverage probabilitys and average rate under co-channel deployment of a two-tier network in terms of exclusion range (\mbox{\small{$r_m$}}), fractional load (\mbox{\small{$\lambda_M$}}) and FAP density (\mbox{\small{$\lambda_F$}}). Note that the macro user density $\lambda_M$ is considered as load parameters. Finally, we provide comprehensive analysis of network deployment in co-channel, orthogonal, and partial modes and also comment on their suitability for a given set of network parameters.} For numerical analysis, the parameters are considered to be \mbox{\small{$\alpha=4$}}, \mbox{\small{$\beta_M=\beta_F=1$}}, \mbox{\small{$N=20$}}, \mbox{\small{$\lambda_B=2e{-5}$}} meter$^{-2}$,  \mbox{\small{$c=7.85$, $R=50$}}, \mbox{\small{$r_0=15$}}, \mbox{\small{$\frac{R_{\text{th}}}{B}=0.5$}} b/s/Hz, and \mbox{\small{$\mu=1$}} per min; unless otherwise mentioned. 
\subsection{Coverage Probability and Activity Factor Validation}
Fig. \ref{fig:SimulationValidation} demonstrate a close agreement between the simulations and derived coverage probability for both MU and FU. It can be seen that the activity factor \mbox{\small{$\zeta_C$}} is the pivotal element of coverage probability for both MU and FU. Coverage of an MU further depends on the FAP density \mbox{\small{$\lambda_F$}} as it characterizes the inter-tier interference. Coverage probabilities of MUs and FUs also depends  on the choice of exclusion range ($r_m$). It is evident from Fig. \ref{fig:SimulationValidation} that the coverage of an MU (FU) is improves (degrades) with increase in $r_m$.  
Deviation in coverage probability of an FU is attributed to the fact that the assumption of fixed distance of FU from FAP made for analytical tractability. Use of mean distance in analytical framework causes mismatch in higher moments of SIR distribution. This justifies the reason of performance gap for very large and small values of SIR threshold ($\beta_F$). The close agreement between analytical and simulation results can be obtained by considering the FU distribution with a non-singular path loss model, such as $(1+r^\alpha)^{-1}$, in the analysis. However, this may hamper the analytical tractability.

\begin{figure}
\centering
\includegraphics[width=.5\textwidth]{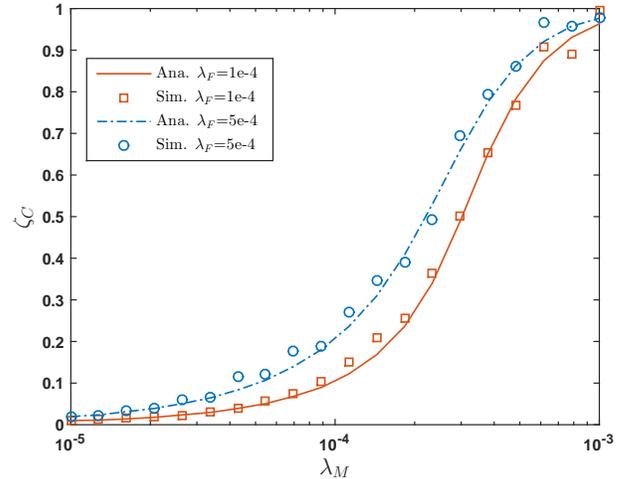} 
\caption{Validation of the Activity Factor of an MBS for \mbox{\small{$r_m=60$}}.}
\label{fig:ActivityFactorValidation}
\end{figure}
Fig. \ref{fig:ActivityFactorValidation} validates the designed framework for the activity factor of an MBS through simulations for co-channel mode. Higher activity of an MBS is seen with the increase of FAPs density for given traffic load. This can be attributed to the fact of increased interference because of increased FAP density causes to consume more bandwidth per service. The effective density of interfering FAP can be controlled be parameter $r_m$ (i.e. exclusion range), for example  \mbox{\small{$r_m=\infty$}} implies a zero femto-tier interference. It can be seen that the increase in activity factor due to increase of FAP density is higher in the middle portion of the load i.e. $\lambda_M$. This is mainly due to the fact that in lower region of $\lambda_M$ the number of MUs in a macro cell are inadequate. Therefore, even though the femto-tier interference is higher, the effective contribution in the increment of activity factor with lesser MUs is smaller compared to that with higher number of MUs. Moreover at higher values of $\lambda_M$, the probability of channel access (\mbox{\small{$p_{r_m}$}}) for femto-tier reduces leading to reduced overall interference suffered by MU. 
\subsection{Performance under Co-Channel Mode}

\begin{figure*}[!t]
\hspace{-.5cm}
\subfigure[]{
    \includegraphics[width=.36\textwidth]{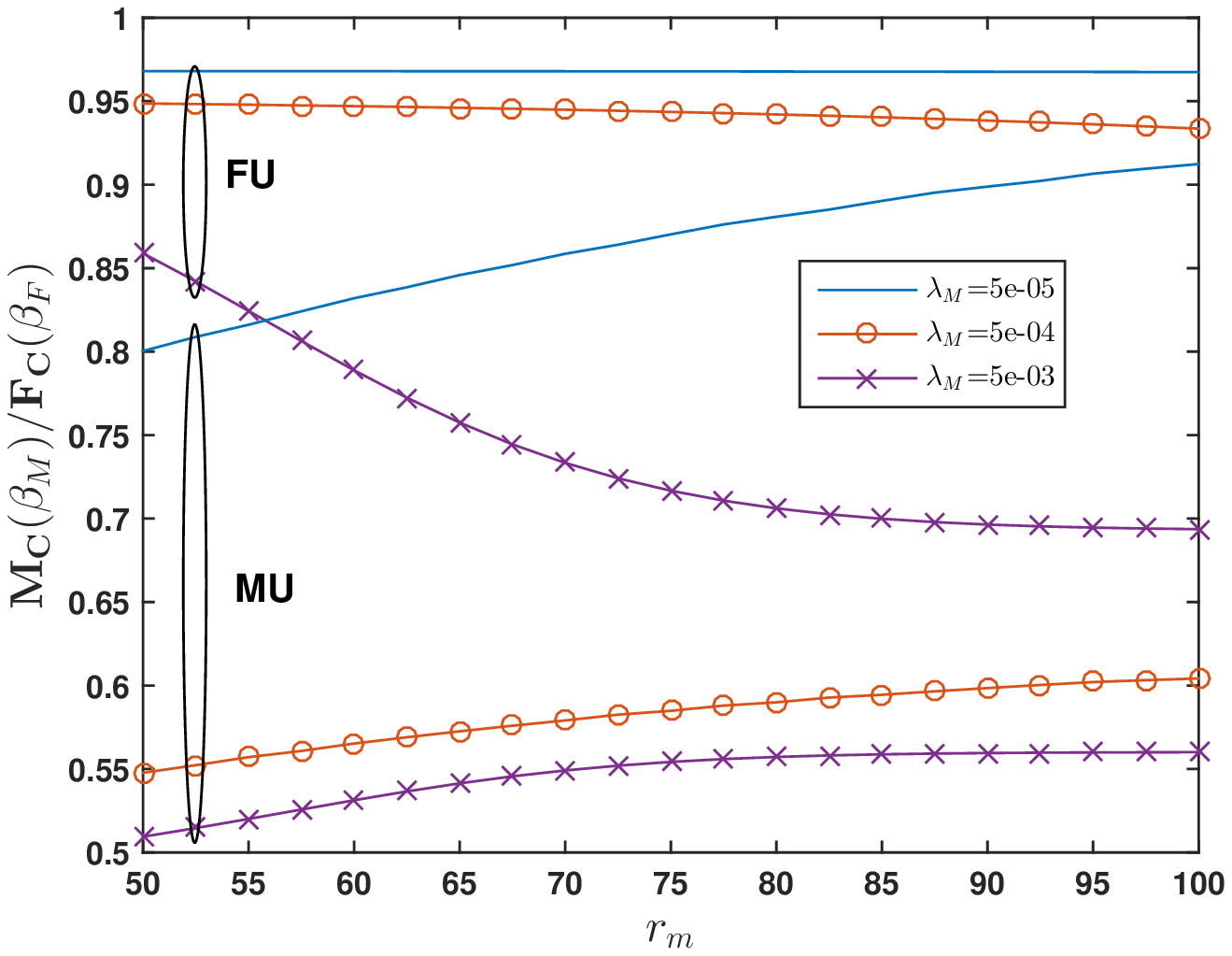} 
    \label{fig:OutP_vs_dm}
}\hspace{-.75cm}
\subfigure[]{
    \includegraphics[width=.36\textwidth]{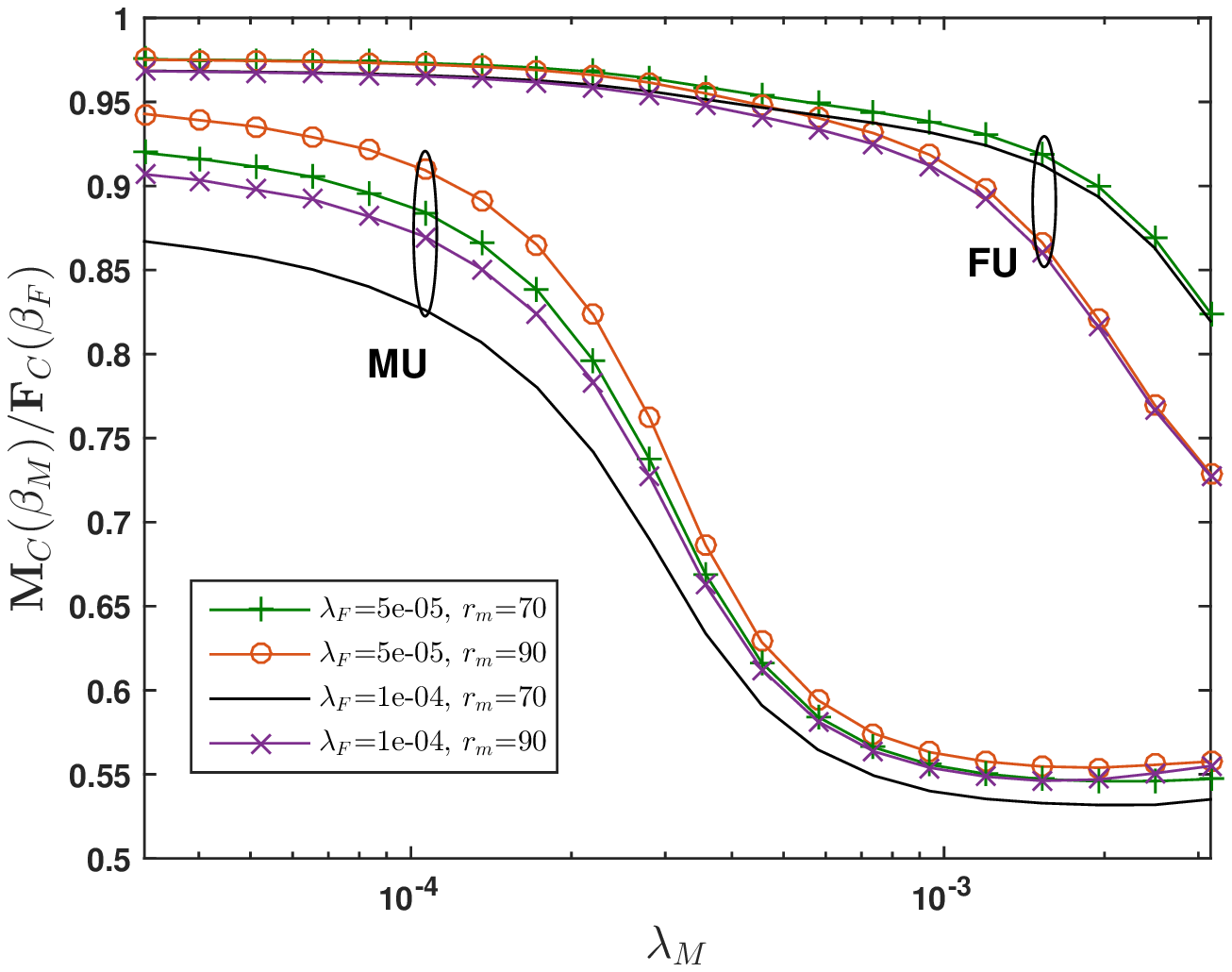} 
    \label{fig:OutP_vs_lamM}
}\hspace{-.75cm}
\subfigure[]{
    \includegraphics[width=.36\textwidth]{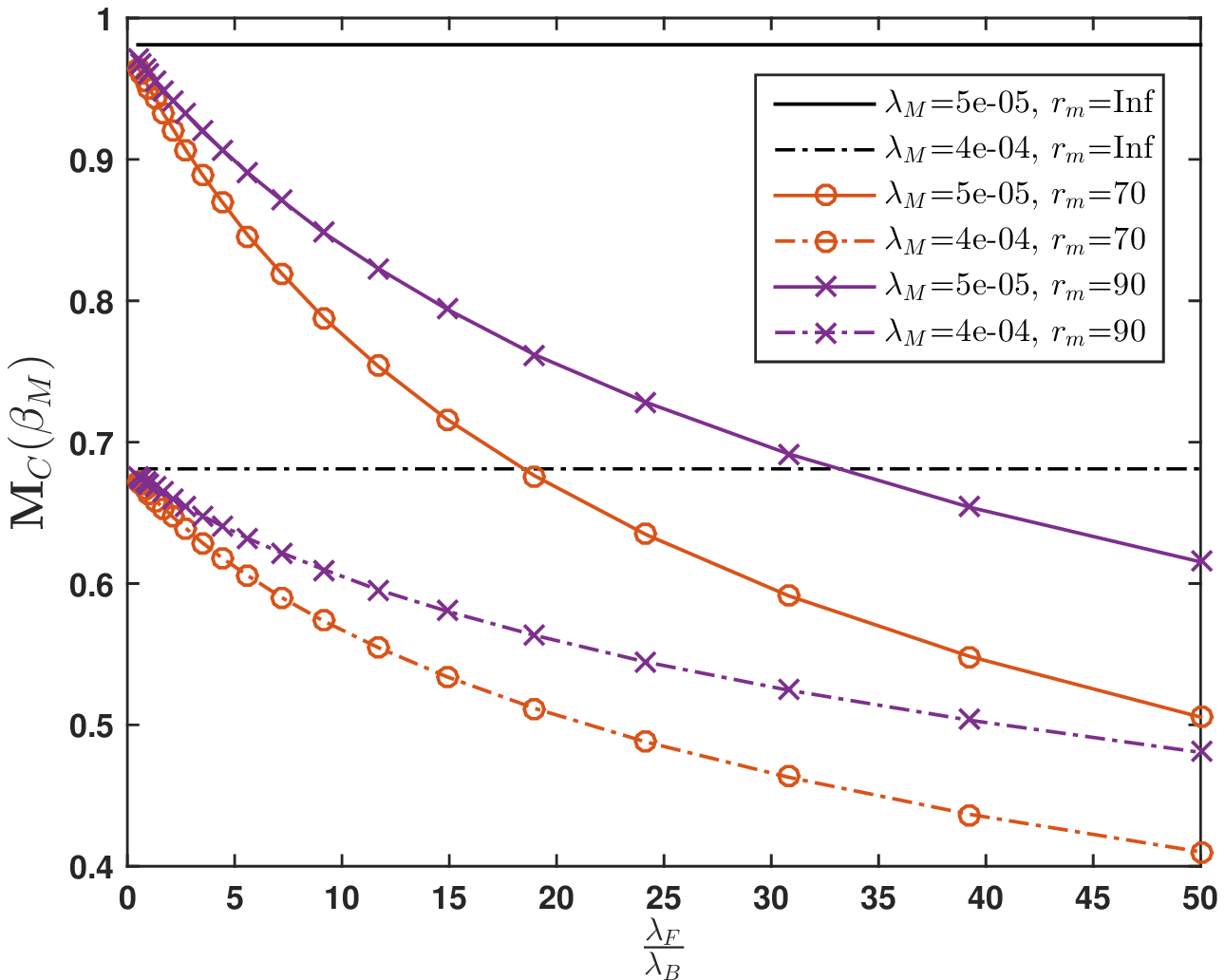} 
    \label{fig:OutP_vs_lamFbylamB} 
}\vspace{-.35cm}
\caption{Coverage probability of MU and FU Vs. \ref{fig:OutP_vs_dm} \mbox{\small{$r_m$}} for \mbox{\small{$\lambda_F=1\text{e}{-4}$}}, \ref{fig:OutP_vs_lamM} \mbox{\small{$\lambda_M$}}, and \ref{fig:OutP_vs_lamFbylamB} \mbox{\small{$\frac{\lambda_F}{\lambda_B}$}}.} 
\label{fig:OutP_vs_dm_and_lamFByLamB}\vspace{-.35cm}
\end{figure*}
\begin{figure}[!b]
\vspace{-.7cm}
 \includegraphics[width=.5\textwidth]{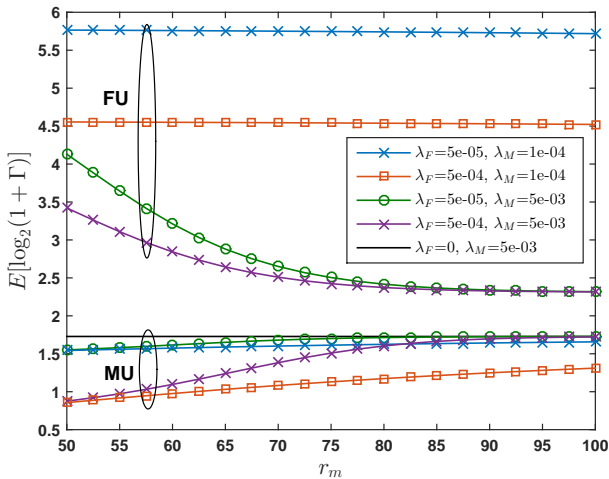} 
\caption{Average rate for best-effort traffic Vs. $r_m$ with \small{$\lambda_B=1\text{e}{-5}$}.}
\label{fig:AvgRateBestEffort_vs_dm} 
\end{figure}
\begin{figure*}[!b]
\centering
\subfigure[]{
    \centering
    \includegraphics[width=.4\textwidth]{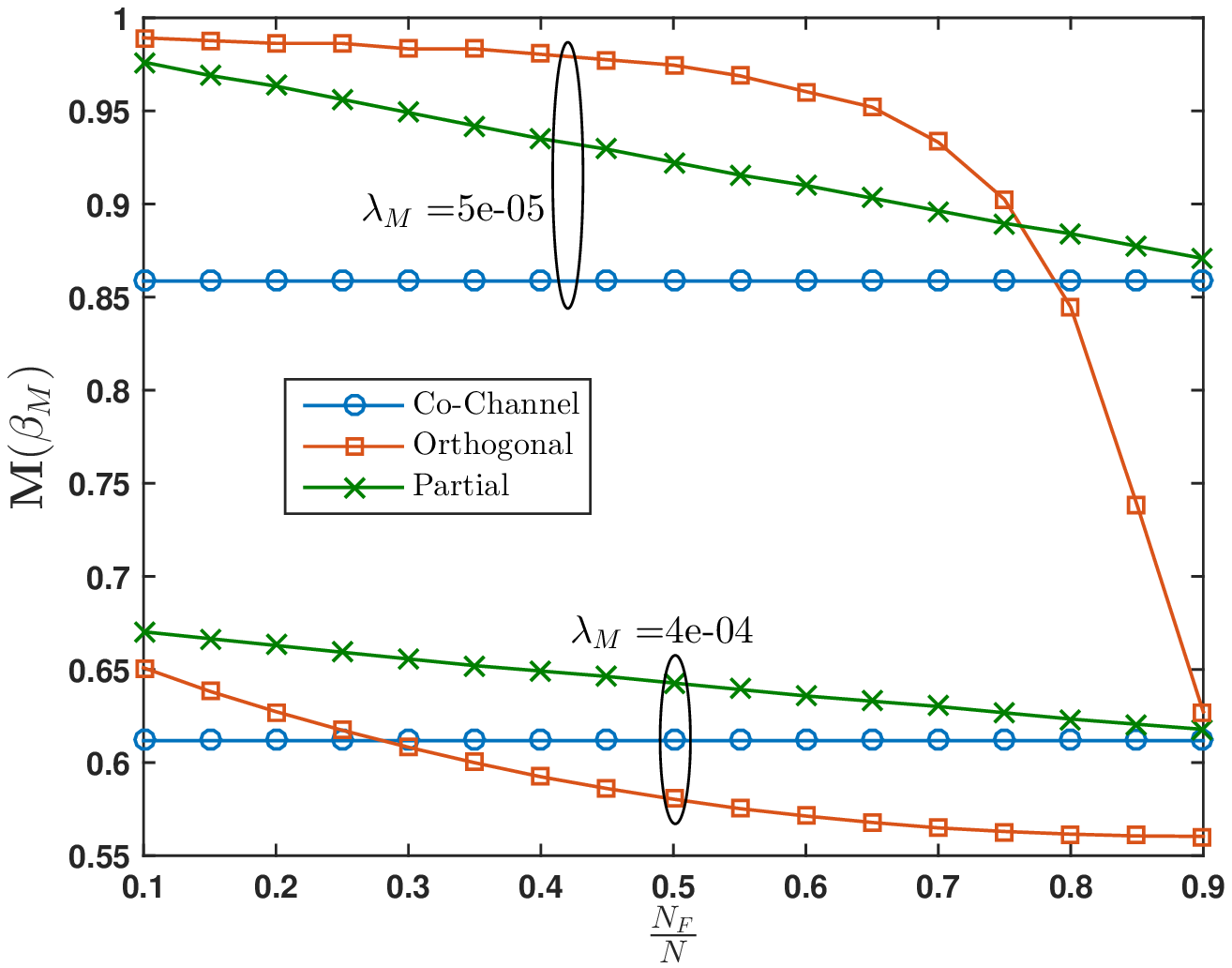} 
    \label{fig:coverage probabilityMacroUser_Comparison}
}
\subfigure[]{
    \centering
    \includegraphics[width=.4\textwidth]{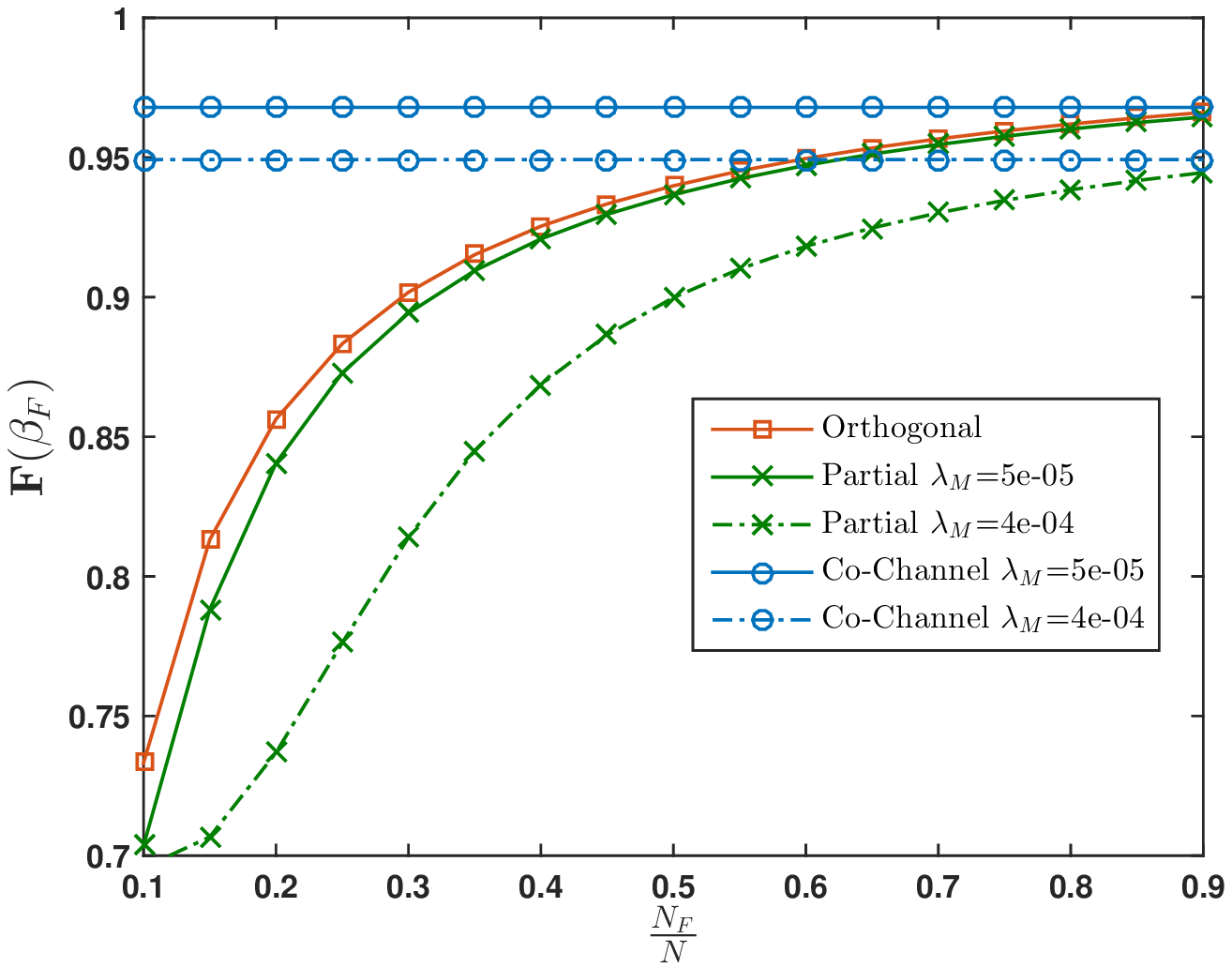} 
    \label{fig:coverage probabilityfemtoUser_Comparison} 
}\vspace{-.35cm}
\caption{Comparison of co-channel, orthogonal, and partial modes for coverage of an MU in \ref{fig:coverage probabilityMacroUser_Comparison} and an FU in \ref{fig:coverage probabilityfemtoUser_Comparison} Vs. \mbox{\small{$\frac{N_F}{N}$}} with \mbox{\small{$\lambda_F=1\text{e}{-4}$}} and \mbox{\small{$r_m=70$}}.} 
\label{fig:coverage probability_Comparison}
\end{figure*}
\begin{figure*}[!t]
\vspace{-.35cm}
\centering
\subfigure[]{
    \centering
     \includegraphics[width=.4\textwidth]{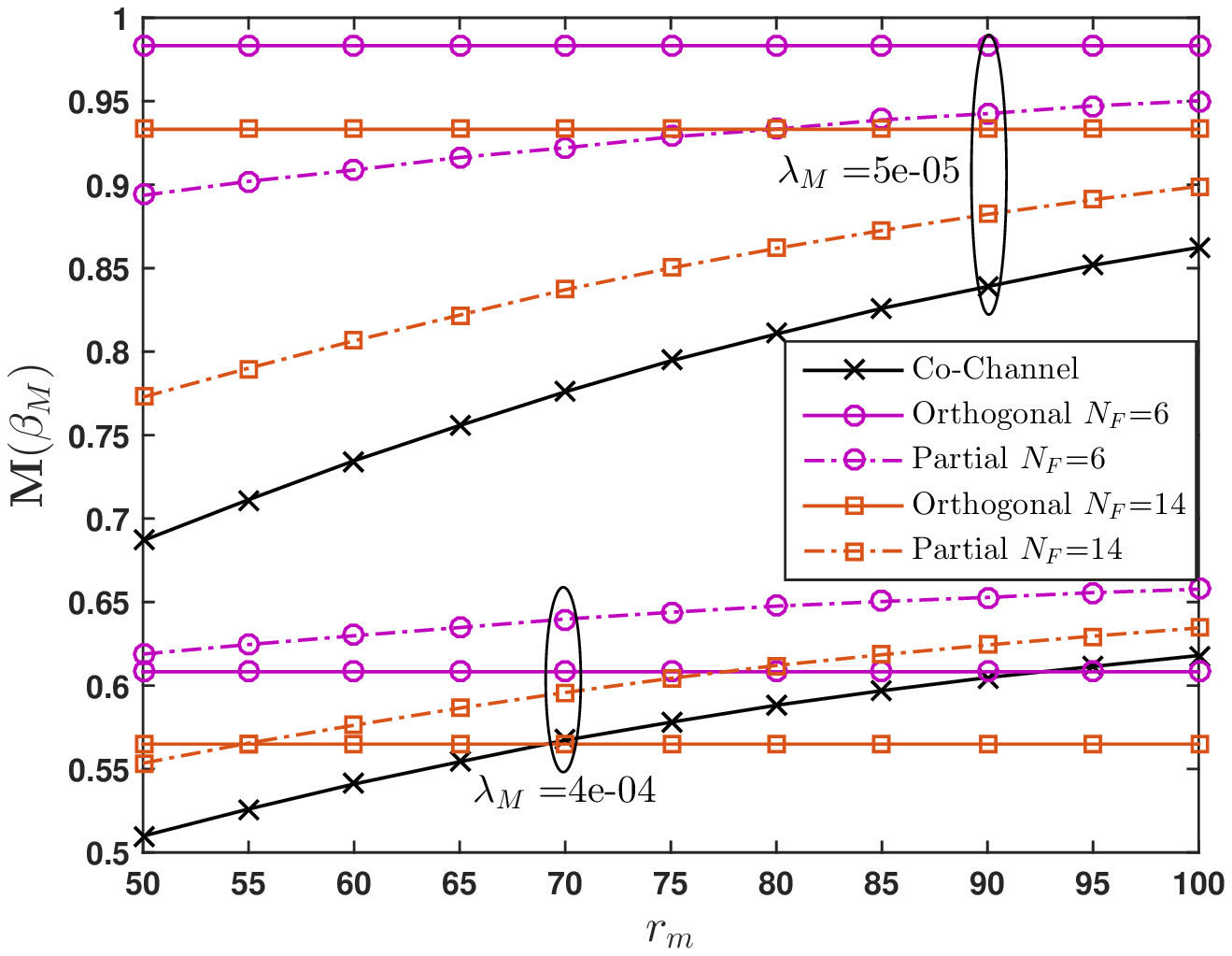} 
    \label{fig:coverage probabilityVsdm_MacroUser_Comparison} 
}
\subfigure[]{
    \centering
     \includegraphics[width=.4\textwidth]{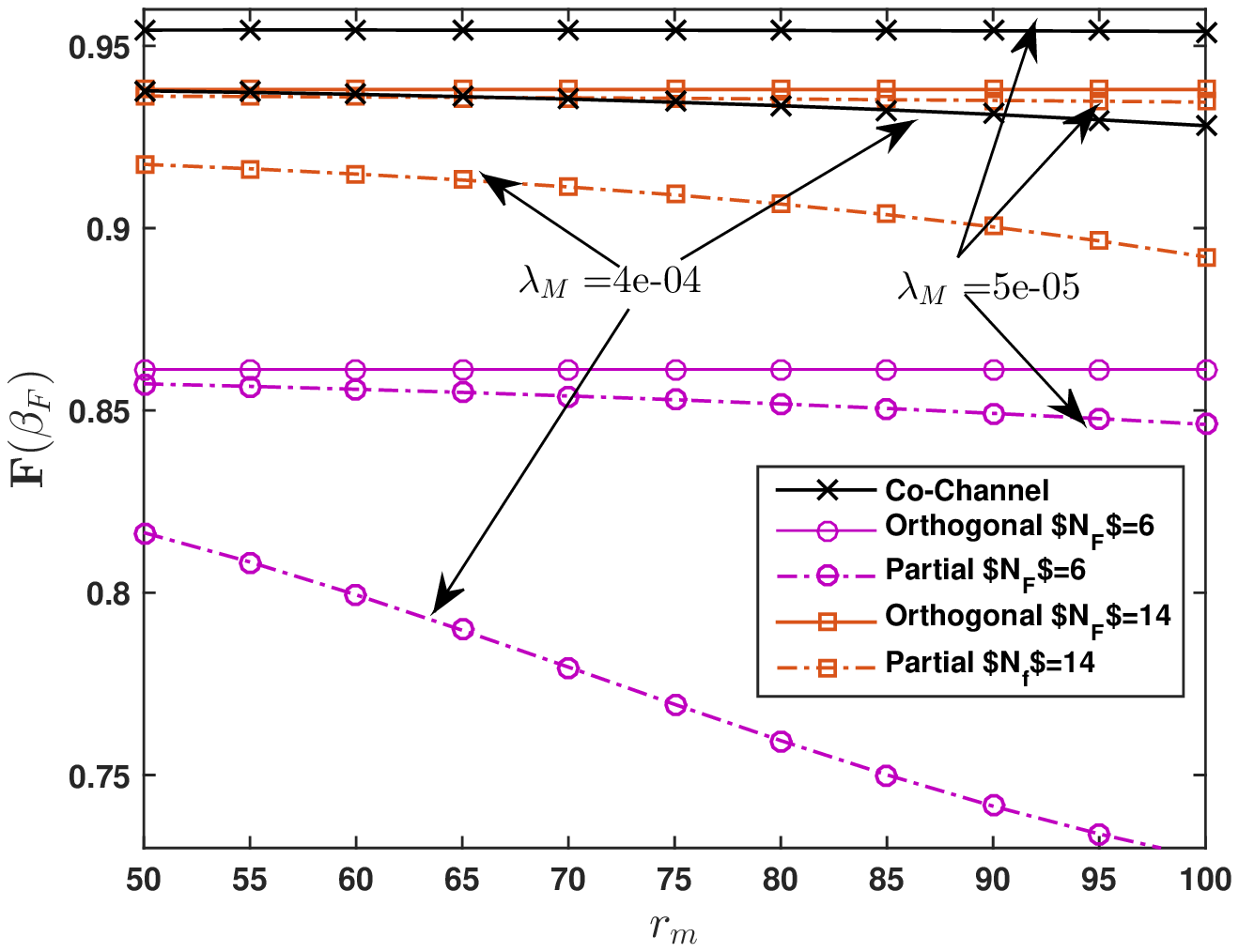} 
    \label{fig:coverage probabilityVsdm_FemtoUser_Comparison} 
}\vspace{-.35cm}
\caption{Comparison of co-channel, orthogonal, and partial modes for coverage of an MU in \ref{fig:coverage probabilityVsdm_MacroUser_Comparison}, and an FU in \ref{fig:coverage probabilityVsdm_FemtoUser_Comparison} Vs. \mbox{\small{$r_m$}} with \mbox{\small{$\lambda_F=1\text{e}{-4}$}}.}
\label{fig:coverage probabilityVsdm_Comparison} \vspace{-.35cm}
\end{figure*}
Fig. \ref{fig:OutP_vs_dm} depicts the coverage probability of an MU and an FU versus \mbox{\small{$r_m$}}. As expected the coverage of an MU improves with the increasing value of \mbox{\small{$r_m$}}. However, it increases at a faster rate in the lower traffic intensity of macro-tier (i.e. lower values of \mbox{\small{$\lambda_M$}}), since the intra-tier interference is less for a lower MBS activity. Moreover, the saturation is seen at a higher \mbox{\small{$\lambda_M$}} due to the fact that \mbox{\small{$\zeta_C\approx 1$}} and \mbox{\small{$p_{r_m}\ll 1$}}. The coverage of FU is more or less unaffected by \mbox{\small{$r_m$}} for lower macro traffic intensity. However, it degrades with \mbox{\small{$r_m$}} for higher values of \mbox{\small{$\lambda_M$}} as the intra-cluster interference increases due to significantly dropped $p_{r_m}$ at higher values of $\lambda_M$.
Next, Fig. \ref{fig:OutP_vs_lamM} depicts the coverage of an MU and an FU versus MU traffic load (\mbox{\small{$\lambda_M$}}). The coverage degrading for MUs and FUs with increase in the value of \mbox{\small{$\lambda_M$}} can be observed as the activity factor of MBS increases with \mbox{\small{$\lambda_M$}}. However, the rate of degradation depends on $r_m$ and $\lambda_F$ which decides the femto-tier interference for an MU and thus the activity factor of an MBS. With further increase in \mbox{\small{$\lambda_M$}}, a saturation effect is observed in coverage of MU as the activity factor of an MBS reaches to one and the femto-tier interference is minimal at a higher \mbox{\small{$\lambda_M$}}. 
From Fig. \ref{fig:OutP_vs_lamFbylamB} it can be seen that coverage of an MU drops with \mbox{\small{$\lambda_F$}}. The starting point of drop depends on the macro-tier traffic intensity. 

Fig. \ref{fig:AvgRateBestEffort_vs_dm} shows the average rate of an MU and an FU in a best-effort traffic environment which is dependent on the \mbox{\small{$\lambda_M$}} even though the activity factor of MBS is set to one. This is mainly because the channel accessibility of FAPs is \mbox{\small{$\lambda_M$}} dependent which decides the femto-tier interference. It can be noted that the average rate of an MU increases and an FU decreases with the increase in \mbox{\small{$r_m$}}. However, slope of increment and decrement is conditioned on \mbox{\small{$\lambda_M$}}.  
\subsection{Comparative Analysis of Coverage in Co-Channel, Orthogonal, and Partial Spectrum Sharing Modes}
Here we present the comparison of coverage probabilities of an MU and an FU experienced in co-channel, orthogonal, and partial spectrum sharing modes in terms of parameters \mbox{\small{$r_m$}} and \mbox{\small{$N_F$}} for a given load (\mbox{\small{$\lambda_M$}}) and femto cluster density (\mbox{\small{$\lambda_F$}}).

Fig. \ref{fig:coverage probability_Comparison} depicts the comparison of coverage probabilities of an MU and an FU versus \mbox{\small{$N_F$}}. It can be seen that the coverage of an MU (FU) is a non-increasing (non-decreasing) function of \mbox{\small{$N_F$}}. Fig. \ref{fig:coverage probabilityMacroUser_Comparison} shows that the partial/orthogonal mode yields better coverage for MUs for smaller \mbox{\small{$N_F$}} compared to co-channel mode as it avoids severe femto-tier interference over \mbox{\small{$N-N_F$}} number of channels. At higher \mbox{\small{$N_F$}}, however, co-channel mode yields better coverage for MUs compared to orthogonal mode. On the contrary, the co-channel mode always provides lower coverage compared to partial mode. Moreover, it can be seen that the effectiveness of partial/orthogonal mode over co-channel mode is higher in lower load scenario as \mbox{\small{$\lambda_M$}} has direct impact on the MBS activity which further decides the intra-tier interference. Fig. \ref{fig:coverage probabilityfemtoUser_Comparison} shows that partial/orthogonal mode render reduced coverage for FUs compared to the co-channel  mode. This is due to the fact that: 1) these modes limit the span of channel access for FUs which leads to increased intra-tier interference, and 2) interference imposed by MBSs on the shorter femto links is significantly smaller compared to their intra-tier interference.
 
Fig. \ref{fig:coverage probabilityVsdm_Comparison} shows that the coverage experienced by an MU (FU) is a non-decreasing (non-increasing) function of \mbox{\small{$r_m$}}. This indicates that the variations of coverage probabilities w.r.t \mbox{\small{$r_m$}} are opposite to that of w.r.t \mbox{\small{$N_F$}}. Note that the coverage in orthogonal mode is independent of \mbox{\small{$r_m$}}. Fig. \ref{fig:coverage probabilityVsdm_MacroUser_Comparison} depicts that the coverage of an MU increases with \mbox{\small{$r_m$}} under co-channel/partial mode and moves beyond the value of coverage that is realized using orthogonal mode. It can be observed that the disparity in the coverage under these modes is relatively more for lesser load compared to that in higher load scenario. This is due to a lower activity of an MBS. From Fig. \ref{fig:coverage probabilityVsdm_FemtoUser_Comparison} it can be seen that in a co-channel/partial mode the coverage of an FU drops with increase in \mbox{\small{$r_m$}}. This is because, at higher values of \mbox{\small{$r_m$}} the channel availability for a cluster is reduced which further increases the intra-cluster interference. 

Moreover, Fig. \ref{fig:coverage probabilityVsLamF_Comparison} depicts the coverage of an MU versus femto cluster density (\mbox{\small{$\lambda_F$}}). As expected, the coverage reduces with increase in \mbox{\small{$\lambda_F$}}. Note that the coverage in orthogonal mode is independent of \mbox{\small{$\lambda_F$}}.  It can be seen that the co-channel and partial mode provide better coverage for smaller values of \mbox{\small{$\lambda_F$}} compared to that of orthogonal mode. However, at larger values of \mbox{\small{$\lambda_F$}} the orthogonal mode becomes better. The crossover point is dependent on \mbox{\small{$N_F$}} and \mbox{\small{$\lambda_M$}}. It is observed that the crossover point shifts towards right with increase in the value of \mbox{\small{$N_F$}}. 
Moreover, the crossover point further shifts towards right with increased \mbox{\small{$\lambda_M$}}. This implies that the orthogonal mode with suitable value of \mbox{\small{$N_F$}} seems to be a natural choice for higher vales of \mbox{\small{$\lambda_F$}}. For the lower load conditions of macro-tier and smaller \mbox{\small{$\lambda_F$}}, the co-channel (or partial with suitable \mbox{\small{$N_F$}} that keeps outage of FU within limit) can provide better coverage to MUs.
\begin{figure}[!h]
\vspace{-.35cm}
\centering
\includegraphics[width=.45\textwidth]{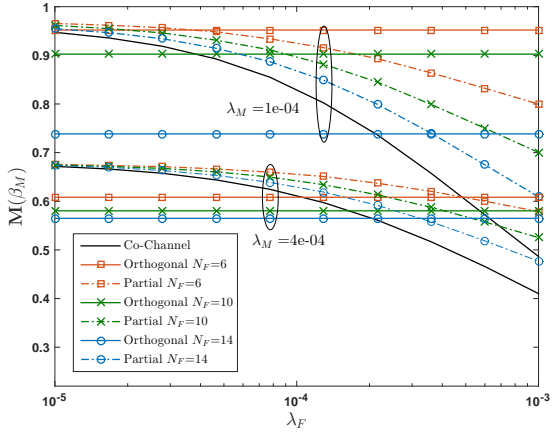} 
\caption{Coverage probability of an MU Vs. \mbox{\small{$\lambda_F$}} for co-channel, orthogonal, and partial modes.}
\label{fig:coverage probabilityVsLamF_Comparison} \vspace{-.5cm}
\end{figure}
\section{Conclusion}
\label{sec:Conclusion}
{In this paper, we present an coverage analysis for a two-tier cellular network under fractional load scenario. We consider the real-time traffic for macro-tier and best-effort traffic for femto-tier under closed access to seize the realistic traffic scenario.} For practical relevance, we model the spatial randomness of macro base stations (MBSs) and femto access points (FAPs) using homogeneous Poisson point process and Poisson cluster process, respectively. To avoid severe inter-tier interference to macro users, we have considered cognitive FAPs which accesses the channels having nearest macro users beyond the exclusion range from the center of a cluster of FAPs. Further the available set of channels are assumed to be equally distributed among the FAPs within the cluster. {Because of the fractional load of macro-tier, the MBS activity factor has role in the coverage analysis. Moreover, coverage probability decides the bandwidth consumption of a service which further decides the activity factor. Therefore, exploiting this coupling, we developed a framework to evaluate the activity factor recursively which aids the evaluation of actual coverage under fractional load conditions.} The derived coverage probability expressions for a femto user and macro users are validated using simulation results. Through numerical results we demonstrate the impact of fractional load and FAP density on the activity factor of an MBS, henceforth on coverage. Next, we provide the comparative analysis of coverage probability in three types of spectrum sharing modes, \text{viz.}: co-channel, orthogonal, and partial. It is shown that the orthogonal mode yields better coverage at higher femto cluster density. Furthermore, we also demonstrate that the reservation of a few number of channels with larger exclusion range for FAPs (as in partial mode), can enhance the coverage of macro users significantly. However, one needs to ensure that the outage of femto user is kept within limit. The achievable gain in coverage is observed to be higher for the lower values of macro-tier load. 
\appendices
\section{Proof of Lemma \ref{lemma1}}
\label{app:lemma1}
The \mbox{\small{$\mathbb{E}_o^![\cdot]$}} represents the expectation with respect to the reduced Palm measure.  It is conditional expectation for point processes, given that there is a point of process at origin but without including the point. Since cluster has a Poisson distribution of points, by Slivnyak's theorem \cite{Stoyan}, we have  \mbox{\small{$\mathbb{E}_o^![\prod_{x\in\phi_c}f(x)]=\mathbb{E}[\prod_{x\in\phi_c}f(x)]$}}. The LT of \mbox{\small{$I_{\text{if}}$}} is given as\begingroup\makeatletter\def\f@size{8}\check@mathfonts
\begin{align}
\mathcal{L}_{I_{\text{if}}}(s)&=\mathbb{E}_{\phi,h,n}\left[\exp\left(-sI_{\text{if}}\right)\right]\nonumber\\ 
&=\mathbb{E}_{\phi,h,n}\left[\prod\nolimits_{x_j\in\tilde\phi_c}\exp\left(-\frac{sh_j\chi P_F}{\|x_j\|^{\alpha}}\right)\right]\nonumber\\ 
&\eqa\mathbb{E}_{\phi,n}\left[\prod\nolimits_{x_j\in\tilde\phi_c}\mathcal{L}_h\left(\frac{s\chi P_F}{\|x_j\|^{\alpha}}\right)\right]\nonumber\\
&\eqb\mathbb{E}_{n}\left[\mathbb{E}_k\left[\left[\int_{\mathbb{R}^2}\mathcal{L}_h\left(\frac{s\chi P_F}{\|x\|^{\alpha}}\right)f(x)dx\right]^{\frac{k}{n+1}}\right]\right]\nonumber\\
&\eqc\mathbb{E}_{n}\left[\exp\left(-c\left(1-\left[\int_{\mathbb{R}^2}\mathcal{L}_h\left(\frac{s\chi P_F}{\|x\|^{\alpha}}\right)f(x)dx\right]^{\frac{1}{n+1}}\right)\right)\right]\nonumber\\
&\leqd\mathbb{E}_{n}\left[\exp\left(-\frac{c}{n+1}\int_{\mathbb{R}^2}\left(1-\mathcal{L}_h\left(\frac{s\chi P_F}{\|x\|^{\alpha}}\right)\right)f(x)dx\right)\right]\nonumber\\
&\approxe\exp\left(-{\tilde c}\int_{\mathbb{R}^2}\left(1-\mathcal{L}_h\left(\frac{s\chi P_F}{\|x\|^{\alpha}}\right)\right)f(x)dx\right)\nonumber\\
&\eqf\exp\left(-\frac{2\tilde c}{R^2}\int_{0}^{R}\frac{r}{1+\frac{1}{s\chi P_F}r^{\alpha}}dr\right) 
\label{eq:lemma11}
\end{align}\endgroup
Step (a) follows from the definition of LT of \mbox{\small{$h$}} and due to the fact of \mbox{\small{$h_j$s}} are i.i.d's. Further, (b) follows as the \mbox{\small{$x_j$s}} are i.i.d's and number of co-channel FAPs are \mbox{\small{$\frac{k}{n+1}$}} where \mbox{\small{$k$}} is the number of FAPs in a cluster and \mbox{\small{$n$}} is the number of accessible channels in the cluster. Further (c) follows from the definition of moment generating function (MGF) of \mbox{\small{$k$}}, i.e. \mbox{\small{$\mathbb{E}(z^k)=\exp(-c (1-z))$}}. Using inequality \mbox{\small{$1-z^{\frac{1}{n}}\geq\frac{1}{n}[1-z]$}} for \mbox{\small{$z\geq 0$}} and \mbox{\small{$n>0$}}, we can write Step (d). The approximation in Step (e) is obtained by taking lower bound (using Jensen's inequality) on the upper bound of the function where \mbox{\small{$\tilde c=c\mathbb{E}\left[\frac{1}{n+1}\right]$}}.
Next Step (f) yields through the substitution of \mbox{\small{$\mathcal{L}_h(s)=\frac{1}{1+s}$}} and a Cartesian to polar co-ordinate conversion. Finally, substitution of \mbox{\small{$\frac{1}{\left(s\chi P_F\right)^{{2}/{\alpha}}}r^{2}=t$}} will complete the proof. 

As \mbox{\small{$n$}} follows binomial distribution \eqref{eq:nChannelAccessibleProb}, the \mbox{\small{$\tilde c$}} becomes:
\begingroup\makeatletter\def\f@size{8}\check@mathfonts
\begin{align}
\tilde c &=c\sum\limits_{n=0}^{N-1} \frac{1}{n+1}\dbinom{N-1}{n} p_{r_m}^{n}\left(1-p_{r_m}\right)^{N-1-n}\nonumber\\
 &=\frac{c}{Np_{r_m}}\sum\limits_{j=1}^{N} \dbinom{N}{j} p_{r_m}^{j}\left(1-p_{r_m}\right)^{N-j}\nonumber\\
 &=\frac{c}{Np_{r_m}}\left[1-\left(1-p_{r_m}\right)^N\right]
 \approxa\frac{c}{N}\exp(\pi\lambda_m r_m^2/\mu)
\label{eq:ExpectationOfBinomial}
\end{align}\endgroup
Equation \eqref{eq:ChannelAccessProb} and fact of \mbox{\small{$z^N\rightarrow 0$}}, for \mbox{\small{$z<1$}} and \mbox{\small{$N>>1$}}, yield  Step (a). The approximation becomes accurate for larger \mbox{\small{$N$}}.
\section{Proof of Lemma \ref{lemma1a}}
\label{app:lemma1a}
Here, \mbox{\small{$\mathbb{E}_o^![\cdot]$}} is conditional expectation given that there is a parent point of process at origin but without including it. Since cluster has a Poisson distribution of points, by Slivnyak's theorem \cite{Stoyan}, we have  \mbox{\small{$\mathbb{E}_o^![\prod_{x\in\Phi_F}f(x)]=\mathbb{E}[\prod_{x\in\Phi_F}f(x)]$}}. 
The probability generating functional (PGFL) of PCP is \cite{Stoyan} 

\begingroup\makeatletter\def\f@size{7}\check@mathfonts
\begin{align}
\mathbb{E}\left[\prod_{x\in\Phi}v(x)\right]=\exp\left(-\lambda\int_{\mathbb{R}^2}1-M\left(\int_{\mathbb{R}^2}v(x+y)f(x)dx\right)dy\right) 
\label{eq:probability_generating_function_PCP}
\end{align}\endgroup
where \mbox{\small{$M(z)=\exp\left(-c\left(1-z\right)\right)$}} is the MGF of Poisson distributed random variable.
The LT of \mbox{\small{$I_{\text{of}}$}} is:
\begingroup\makeatletter\def\f@size{8}\check@mathfonts
\begin{align}
&\mathcal{L}_{I_{\text{of}}}(s)=\mathbb{E}_{\tilde\Phi_F,h,n}\left[\exp\left(-sI_{\text{of}}\right)\right]\nonumber\\
&=\mathbb{E}_{\tilde\Phi_F,h,n}\left[\prod\nolimits_{x_j\in\tilde\Phi_F}\exp\left(-\frac{sh_j\chi P_F}{\|x_j\|^{\alpha}}\right)\right]\nonumber\\
&\eqa\mathbb{E}_{\tilde\Phi_F,n}\left[\prod\nolimits_{x_j\in\tilde\Phi_F}\mathcal{L}_h\left(\frac{s\chi P_F}{\|x_j\|^{\alpha}}\right)\right]\nonumber\\
\begin{split}
&\eqb\mathbb{E}_n\left[\exp\left(-\tilde\lambda_F\int\nolimits_{\mathbb{R}^2}1-\exp\left(-c\left(1-\vphantom{\mathcal{L}_h\left(\frac{x\chi P_F}{\|x-y\|^\alpha}\right)}\right.\right.\right.\right.\\
&\left.\left.\left.\left.~~~~~~~~\left[\int\nolimits_{\mathbb{R}^2}\mathcal{L}_h\left(\frac{s\chi P_F}{\|x-y\|^\alpha}\right)f(x)dx\right]^{\frac{1}{n+1}}\right)\right)dy\right)\right]\nonumber\\
\end{split}\\
\begin{split}
&\leqc\mathbb{E}_n\left[\exp\left(-\tilde\lambda_F\int\nolimits_{\mathbb{R}^2}1-\exp\left(-\frac{c}{n+1}\vphantom{\mathcal{L}_h\left(\frac{x\chi P_F}{\|x-y\|^\alpha}\right)}\right.\right.\right.\\
&\left.\left.\left.~~~~~~~~\int\nolimits_{\mathbb{R}^2}\left[1-\mathcal{L}_h\left(\frac{s\chi P_F}{\|x-y\|^\alpha}\right)\right]f(x)dx\right)dy\right)\right]\nonumber\\
\end{split}\\
&\leqd\mathbb{E}_n\left[\exp\left(-\tilde\lambda_F\frac{c}{n+1}\int\nolimits_{\mathbb{R}^2}\int\nolimits_{\mathbb{R}^2}\left[1-\mathcal{L}_h\left(\frac{s\chi P_F}{\|x-y\|^\alpha}\right)\right]f(x)dxdy\right)\right]\nonumber\\
&\approxe\exp\left(-\tilde\lambda_F\tilde c\int\nolimits_{\mathbb{R}^2}\int\nolimits_{\mathbb{R}^2}\left[1-\mathcal{L}_h\left(\frac{s\chi P_F}{\|x-y\|^\alpha}\right)\right]f(x)dxdy\right)\nonumber\\
&\eqf\exp\left(-\tilde\lambda_F\tilde c\int_{\mathbb{R}^2}f(x)dx \int\nolimits_{\mathbb{R}^2}\frac{dz}{1+\frac{1}{s\chi P_F}\|z\|^{\alpha}}\right)
\label{eq:LT_PCP_InterCluster1}
\end{align}\endgroup
where \mbox{\small{$\tilde\lambda_F=\lambda_Fp_{r_m}$}} is the thinned density of femto cluster.
Step (a) follows from the definition of LT of \mbox{\small{$h$}} and due to the fact of \mbox{\small{$h_j$s}} are i.i.d's.
Step (b) is followed by applying PGFL of PCP process. Using inequality \mbox{\small{$1-z^{\frac{1}{n}}\geq\frac{1}{n}[1-z]$}} for \mbox{\small{$z\geq 0$}} and \mbox{\small{$n>0$}}, we can write Step (c). 
Inequality in Step (d) is followed using \mbox{\small{$1-\exp(-\Delta x)\leq \Delta x$}} for \mbox{\small{$\Delta\geq 0$}}.
Next, the approximation in Step (e) is obtained by taking lower bound (using Jensen's inequality) on upper bound of the function where \mbox{\small{$\tilde c=c\mathbb{E}\left[\frac{1}{n+1}\right]$}}.
Step (f) is followed by substituting \mbox{\small{$\mathcal{L}_h(s)=\frac{1}{1+s}$}} and replacing \mbox{\small{$x-y$}} by \mbox{\small{$z$}}.
Further solving \eqref{eq:LT_PCP_InterCluster1} yields the \eqref{eq:Laplace_InterClusterInterference}.  
\section{Proof of Lemma \ref{lemma3}}
\label{app:lemma3}
Interference from co-channel FAPs to a macro user is \mbox{\small{$I_{\text{fm}}=\sum_{x_i\in\tilde\Phi_F'}h_i\|x_i\|^{-\alpha}\chi {P_F}$}}
where parents point of \mbox{\small{$\tilde\Phi_F$}} has 0 intensity inside \mbox{\small{$\mathcal{B}\left(0,r_m\right)$}} and \mbox{\small{$\tilde\lambda_F$}} outside \mbox{\small{$\mathcal{B}\left(0,r_m\right)$}}. 
Let $\mathcal{B}^c_{r_m}$ denotes the complement of a ball of radius $r_m$ centered at origin, i.e. $\mathbb{R}^2\setminus\mathcal{B}(0,r_m)$. 
Therefore, following the proof of Lemma \ref{lemma1a} till Step (e), we can write\begingroup\makeatletter\def\f@size{8}\check@mathfonts
\begin{align}
\mathcal{L}_{I_{\text{fm}}}\left(s,r_m\right)&\approx\exp\left(-\tilde\lambda_F\tilde c\int\nolimits_{\mathcal{B}^c_{r_m}}\int\nolimits_{\mathbb{R}^2}\left[1-\mathcal{L}_h\left(\frac{s\chi P_F}{\|x-y\|^\alpha}\right)\right]f(x)dxdy\right)\nonumber\\
\label{eq:Laplace_PCPInterference2}
\end{align}\endgroup
Further, substituting \mbox{\small{$\mathcal{L}_h(s)=\frac{1}{1+s}$}} in \eqref{eq:Laplace_PCPInterference2} yields \eqref{eq:Laplace_PCPInterference1}.
For $r_m=0$ we have $\mathcal{B}^c_{r_m}=\mathbb{R}^2$. Therefore, substituting \mbox{\small{$\mathcal{L}_h(s)=\frac{1}{1+s}$}} and replacing $x-y$ by $z$ we can rewrite \eqref{eq:Laplace_PCPInterference2} as follows
\begingroup\makeatletter\def\f@size{8}\check@mathfonts
\begin{align} 
\mathcal{L}_{I_{\text{fm}}}\left(s,r_m\right)&\approx\exp\left(-\tilde\lambda_F\tilde c \int\nolimits_{\mathbb{R}^2} \frac{dz}{1+\frac{1}{s\chi P_F}{\|z\|^\alpha}}dz\right)
\end{align}\endgroup
Further, solving the integral yields \eqref{eq:Laplace_PCPInterference_NonCog}. This completes the proof.

\begin{IEEEbiography}[{\includegraphics[width=1in,height=1.25in,clip,keepaspectratio]{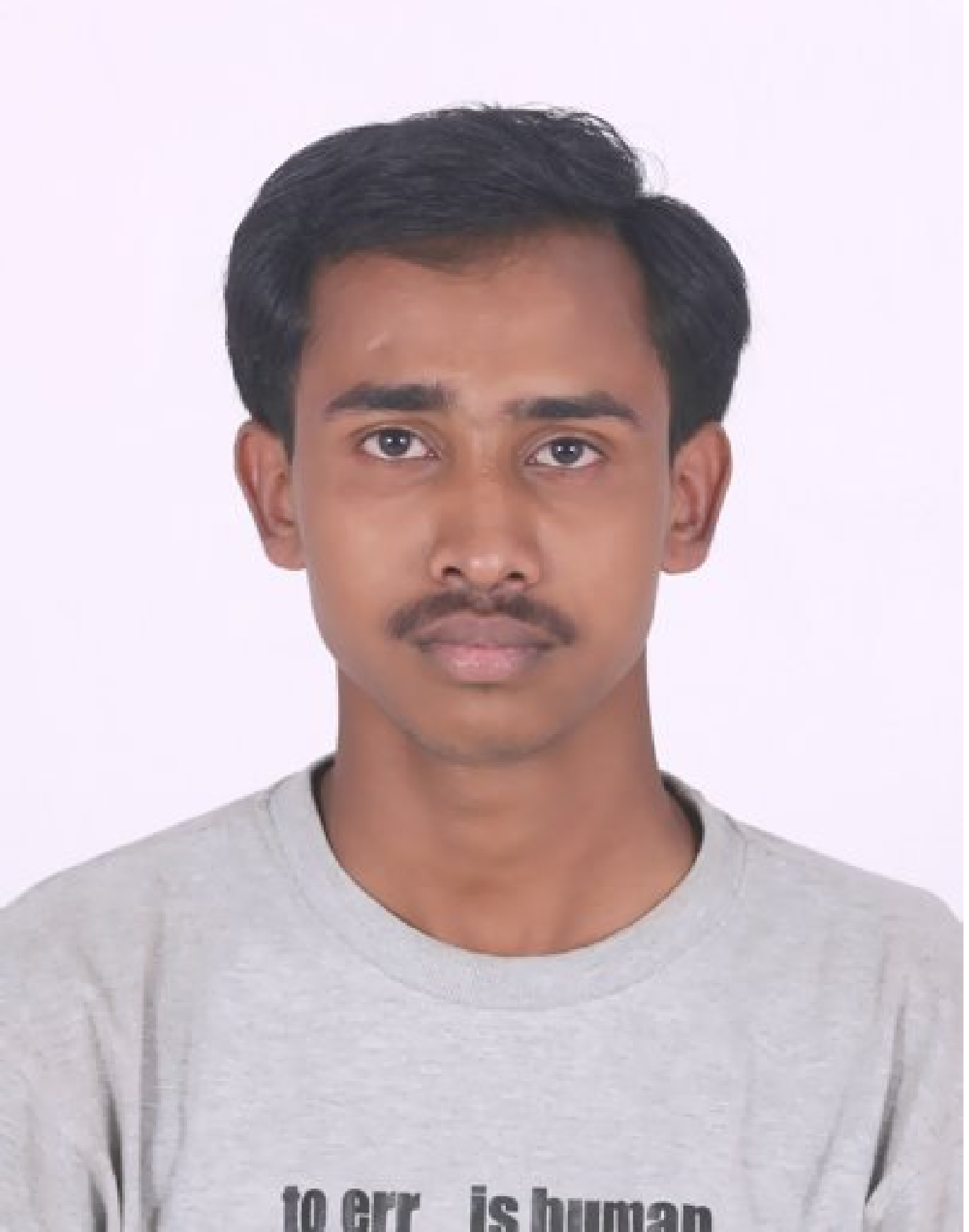}}]{Praful D. Mankar} has
completed B.Tech degree in Electronics and Telecommunications Engineering from Amaravati University, MH, India in 2006. In 2009, he has obtained M.Tech degree in Telecommunication System Engineering  from Indian Institute of Technology (IIT) Kharagpur, WB, India. He received Ph.D degree in wireless communication from IIT Kharagpur, WB, India in 2016.  Currently, He is working as a research assistant at IIT Kharagpur, WB, India. His research interest includes modeling, analyzing and designing of wireless networks. 
\end{IEEEbiography}

\begin{IEEEbiography}[{\includegraphics[width=1in,height=1.25in,clip,keepaspectratio]{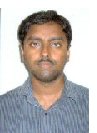}}]{Goutam Das} has obtained his Ph.D. Degree from the University of Melbourne, Australia in 2008. He has worked as a postdoctoral fellow at Ghent University, Belgium, from 2009-2011. Currently, he is working as an Assistant Professor in the Indian Institute of Technology, Kharagpur. He has served as a member in the organizing committee of IEEE ANTS since 2011. His research interests include optical access networks, optical data center networks, radio over fiber technology, optical packet switched networks and media access protocol design for application specific requirements.
\end{IEEEbiography}

\begin{IEEEbiography}[{\includegraphics[width=1in,height=1.25in,clip,keepaspectratio]{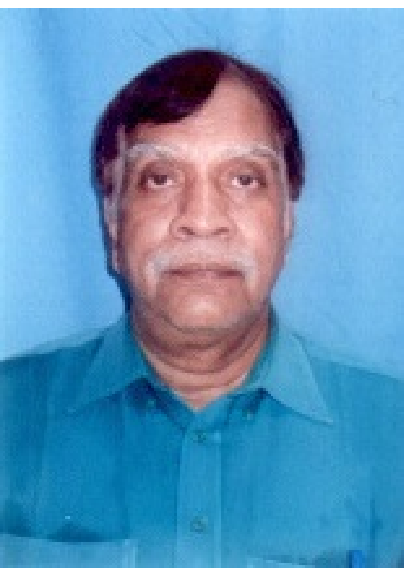}}]{Sant Sharan Pathak} has received his B.Tech and M.Tech Degrees in Electronics Engineering from IT BHU in 1976 and 1978 respectively,
and Ph.D. Degree in Digital Communications from IIT Delhi in 1984. He has joined the Department of Electronics and Electrical
Communication Engineering in 1985. His area of research interest includes physical layer network issues, receiver design optimization for Gaussian and non-Gaussian channels, security system design and analysis at application layer, image forensics with wireless camera pickup over Internet using low bandwidth wireless access channel, and similar others.
\end{IEEEbiography}
\end{document}